\documentclass[format=acmsmall]{acmart}
\usepackage{booktabs} % For formal tables
\usepackage[linesnumbered,lined,boxed,commentsnumbered]{algorithm2e}

\SetAlFnt{\small}
\SetAlCapFnt{\small}
\SetAlCapNameFnt{\small}
\SetAlCapHSkip{0pt}
\IncMargin{-\parindent}

\usepackage{amssymb}
\usepackage{amsmath}
\usepackage{graphicx}
\usepackage{subcaption}
\usepackage{multirow}
\usepackage[english]{babel}

\graphicspath{{./}{./Image/}}

\usepackage{appendix}
\begin{document}
\date{}
%\title{Addressing an Early Version of Unconstrained Via Minimization after Floorplanning}
%\title{Addressing Unconstrained Via Minimization after Floorplanning}
\title{Exploring the Scope of Unconstrained Via Minimization by Recursive Floorplan Bipartitioning}
\author{Bapi Kar}
\affiliation{%
  \institution{Indian Institute of Technology \footnote{currently at Nanyang Technological University, Singapore}}
  \city{Kharagpur}
  \country{India}}
  
\author{Susmita Sur-Kolay}
\affiliation{
  \institution{Indian Statistical Institute}
  \city{Kolkata}
  \country{India}}

\author{Chittaranjan Mandal}
\affiliation{%
  \institution{Indian Institute of Technology}
  \city{Kharagpur}
  \country{India}}

\begin{abstract}
Random via failure is a major concern for post-fabrication reliability and poor manufacturing yield. A demanding solution to this problem is redundant via insertion during post-routing optimization. It becomes very critical when a multi-layer routing solution already incurs a large number of vias. Very few global routers addressed unconstrained via minimization (UVM) problem, while using minimal pattern routing and layer assignment of nets. It also includes a recent floorplan based early global routability assessment tool STAIRoute \cite{karb2}.  

This work addresses an early version of unconstrained via minimization problem during early global routing by identifying a set of minimal bend routing regions in any floorplan, by a new recursive bipartitioning framework. These regions facilitate monotone pattern routing of a set of nets in the floorplan by STAIRoute. The area/number balanced floorplan bipartitionining is a multi-objective optimization problem and known to be NP-hard \cite{majum2}. No existing approaches considered bend minimization as an objective and some of them incurred higher runtime overhead. In this paper, we present a Greedy as well as randomized neighbor search based staircase wave-front propagation methods for obtaining optimal bipartitioning results for minimal bend routing through multiple routing layers, for a balanced trade-off between routability, wirelength and congestion.

Experiments were conducted on MCNC/GSRC floorplanning benchmarks for studying the variation of early via count obtained by STAIRoute for different values of the trade-off parameters ($\gamma, \beta$) in this multi-objective optimization problem, using $8$ metal layers. We studied the impact of ($\gamma, \beta$) values on each of the objectives as well as their linear combination function $Gain$ of these objectives. %A case study on different floorplans of an industrial design, by a popular industrial physical design tool, is also presented here.
\end{abstract}


\maketitle

%% Section: Introduction
\section{Introduction}
With sustained advancement in IC fabrication technology, stringent design rules are evolving by exponentially large numbers. Straightforward routing solutions from a HPWL aware placement solution may not yield an acceptable physical design closure due to too many routing violations in subsequent global routing. If these violations are not resolved by the subsequent detailed routing or by an iterative global and detailed routing, the placement (or floorplanning or even logic restructuring) of the design should be redone. In practice, several iterations in block placement (and floorplanning) are required for complex designs containing multi-million gates in order to attain a feasible routing solution (see Fig. \ref{fig:pdflow} (a)). Therefore, it has been a mandate to consider different global routing metrics such as routability, wirelength, congestion \cite{chang,liuw,tlin} and even timing \cite{sherw} as the prime objectives in a placement problem. Some placement engines, however, integrated faster global routing solutions for iterative improvement of the placement solution \cite{hex,panm1,panm2,viswa,viswa2}. For fewer design iterations before successful routing closure, integrated global/detailed routing methods were also explored \cite{zhang}.
\begin{figure}[!ht]
\centering
\begin{subfigure}[b]{0.48\textwidth}
\centering
\includegraphics[scale=0.65]{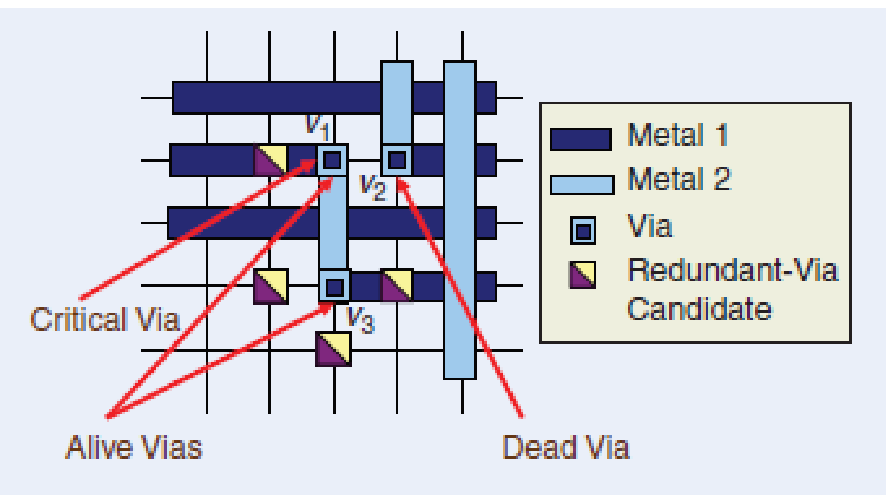}
\caption{}
\end{subfigure}
\begin{subfigure}[b]{0.48\textwidth}
\centering
\includegraphics[scale=0.75]{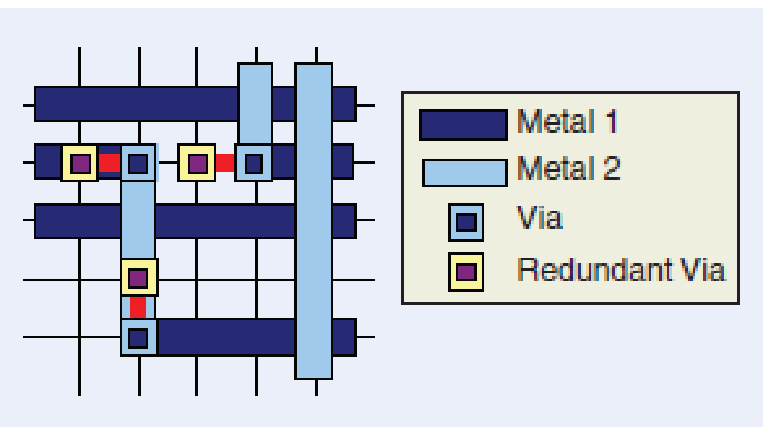}
\caption{}
\end{subfigure}
\begin{subfigure}[b]{\textwidth}
\centering
\includegraphics[scale=0.65]{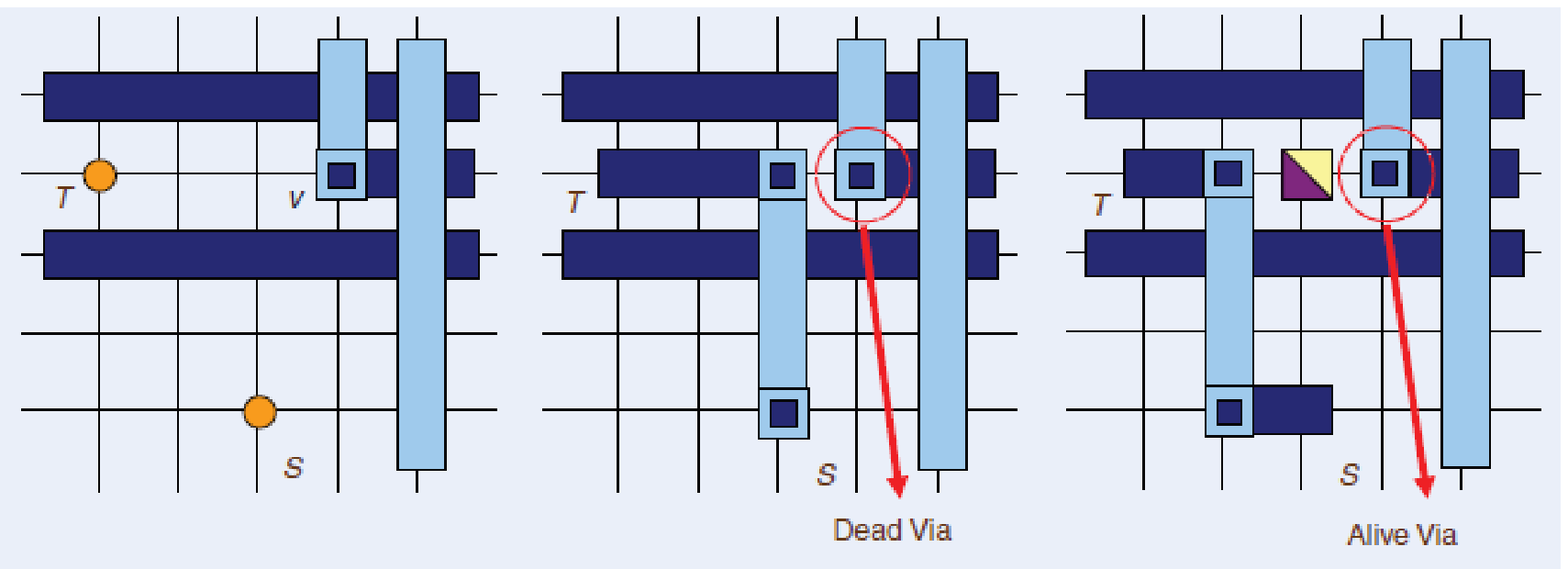}
\caption{}
\end{subfigure}
\caption{Enahancing design for reliability using reduncant-via aware routing \cite{hchen2}}
\label{fig:redun_via}
\end{figure}

Modern VDSM fabrication processes, such as $65nm$ and below, continue to allow more routing layers with varying metal width/pitch for successful routing completion. A routing solution with excessive via count not only causes design for reliability issues due to random via failures \cite{hchen2}, but also impacts the circuit performance due to increased resistance along the routing paths with more vias. Double via insertion during post-routing layout optimization or identifying a via-failure aware routing \cite{hchen2} as depicted in Fig. \ref{fig:redun_via} or even redundant via aware ECO routing during mask optimization \cite{chien} for increased reliability and yield of the fabricated design are some of the known approaches to minimize these failures. Moreover, vias consume substantial routing area and pose as additional routing blockages in the routing regions impacting routability of the design. Therefore, via minimization \cite{sherw} is a critical problem to handle in physical design flow. There are two approaches: (a) unconstrained via minimization (UVM), and (b) constrained via minimization (CVM). While UVM identifies a routing path of a net with minimal number of vias along it for a given number of routing (metal) layers, CVM approaches aims to minimize the number of vias while keeping the routing topology unchanged. This routing topololy is obatined by planar routing solution during early phases of global routing. Although, both are known to be NP-hard problems, UVM is much harder than CVM \cite{sherw,hsu}. Existing global routers \cite{mcho1,panm1,royj,xuy}, except a few like \cite{zcao,luj,msad,zhang2}, used CVM based layer assignment approaches on a planar routing solution for reducing via count as well as mitigating congestion \cite{lee}.
\begin{figure}[!ht]
\centering
\includegraphics[scale=0.45]{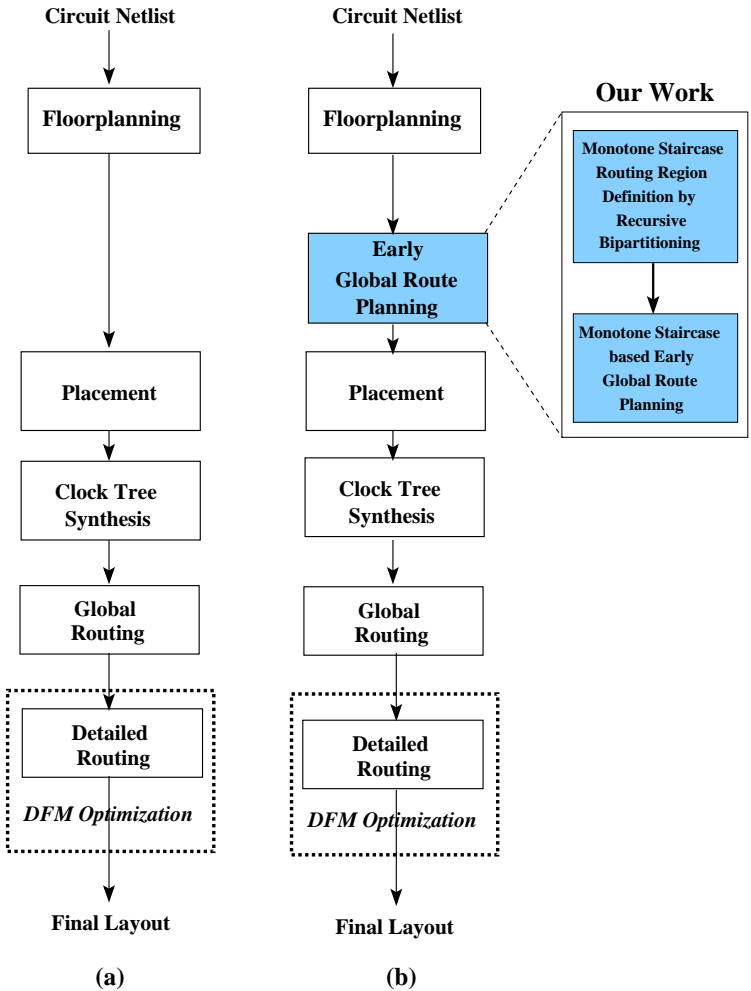}
\caption{Physical Design (PD) Flow: (a) Traditional \cite{sherw,olympus}, and (b) New \cite{karb2}}
\label{fig:pdflow}
\end{figure}

Recently, an early global routing (EGR) method STAIRoute \cite{karb2} was proposed for early routability assessment of a floorplanned layout, facilitated by a monotone staircase cut based recursive floorplan bipartitioning framework \cite{karb,majum1,majum2}. These bipartitioners work on any floorplan irrespective of their sliceability. As \cite{guru,ssk} pointed out, the monotone staircase routing framework ensures an well defined routing order of the nets, based on the net cut information available with the nodes in the bipartitioning hierachy \cite{karb,majum1,majum2}. As highlighted in Fig. \ref{fig:pdflow} (b), STAIRoute works in two stages: (a) enumerating the monotone staircase routing regions in a floorplan by recursive bipartitioning using monotone staircase cuts \cite{karb,majum1,majum2}, and (b) proposing an early global routing model for routing these nets through a number of metal layers, using these bipartitioning results. 

The existing bipartitioning methods using monotone staircase cuts \cite{dasg,karb,majum1,majum2} considered only two objectives: (i) the area (number) of the blocks in each bipartition to be maximized, and (ii) the number of nets cut by a bipartition be minimized. While the former objective is related to the height of the bipartition hierarchy (also known as MSC tree \cite{karb}), minimizing the number of nets being cut has several advantages like: (a) distributing the routing paths of the nets uniformly across the entire layout, (b) reducing routing violations due to congestion hot-spots (congestion $> 100\%$), (c) achieving uniform wire distribution across the layout for minimal variation due to \textit{chemical mechanical polishing} (CMP) process, and (d) minimizing cross-talk effect due to long (global) nets running through the longer staircases, specially those nets corresponding to the upper nodes in the bipartition tree. In global routing, routing path of a net using multi-bend monotone pattern routing and its variants L/Z patterns \cite{zcao,kast} is confined within the net bounding boxes. Therefore, identification of a minimal bend monotone patterns can potentially yield fewer via counts, while L/Z patterns use minimum of one/two vias respectively for minimum layer change.

In this work, we propose a new recursive floorplan bipartitioning framework, for identifying minimal bend monotone staircase routing regions in a floorplan, in order to use fewer vias during early global routing of the nets in the floorplan. The key contributions of this paper are:
\begin{enumerate}
\item define a new objective of bend minimization in the existing multi-objective floorplan bipartitioning problem;
\item propose a greedy method for identifying minimal bend monotone staircase routing regions for early global routing with smaller via count (this is an early approach for \textit{unconstrained via minimization} (UVM); and 
\item introduce a randomized neighbor search technique and staircase wave front propagation approach for exploring a larger solution space of potentially optimal minimal bend monotone staircase regions in a floorplan.%, aimed at further reduction in early via count.
\end{enumerate}

The organization of this paper is as follows: in Section \ref{sec:prelim}, we discuss the background on monotone staircase routing region definition in a floorplan. The proposed floorplan bipartitioning method, for identifying a set of monotone staircases with minimal number of bends for the entire floorplan, is presented in Section \ref{sec:bend}. Section \ref{sec:msc-rand} discusses the basis for an extension of this greedy bipartitioning method and illustrates a new randomized neighbor search technique and the corresponding staircase wave-front propagation approach. Experimental results and relevant discussions are covered in Section \ref{sec:result}, followed by the summary of this work in Section \ref{sec:con}.

% Section: Preliminaries
\section{Background on Monotone Staircase Cuts}
\label{sec:prelim}
Before discussing the proposed recursive floorplan bipartitioning method, we revisit the formulation of an unweighted directed graph $G_b(V_b,E_b)$, namely \textit{block adjacency graph} (BAG) \cite{karb,majum1}, used to define the adjacency relation of a set of $n$ blocks $B = \{b_i\}$ in a given floorplan $F$. The graph $G_b(V_b,E_b)$ is defined as follows: the vertex set $V_b$ = \{$v_i | v_i$ corresponds to block $b_i$\} and the edge set $E_b = \{e_{ij}\}$ where $e_{ij}$ = \{($v_i$,$v_j$) $|$ block $b_i$ is on the \textit{left of (above)} an adjacent block $b_j$ in $F$\}. The vertices corresponding to the top-left and the bottom-right corner blocks are designated as the source and the sink vertices respectively, with zero in-degree and out-degree respectively. This definition yields a \textit{monotonically increasing staircase} (MIS) $C_I$ (see Fig. \ref{fig:bag} (a)).
\begin{figure}[!ht]
\centering
\includegraphics[scale=0.33]{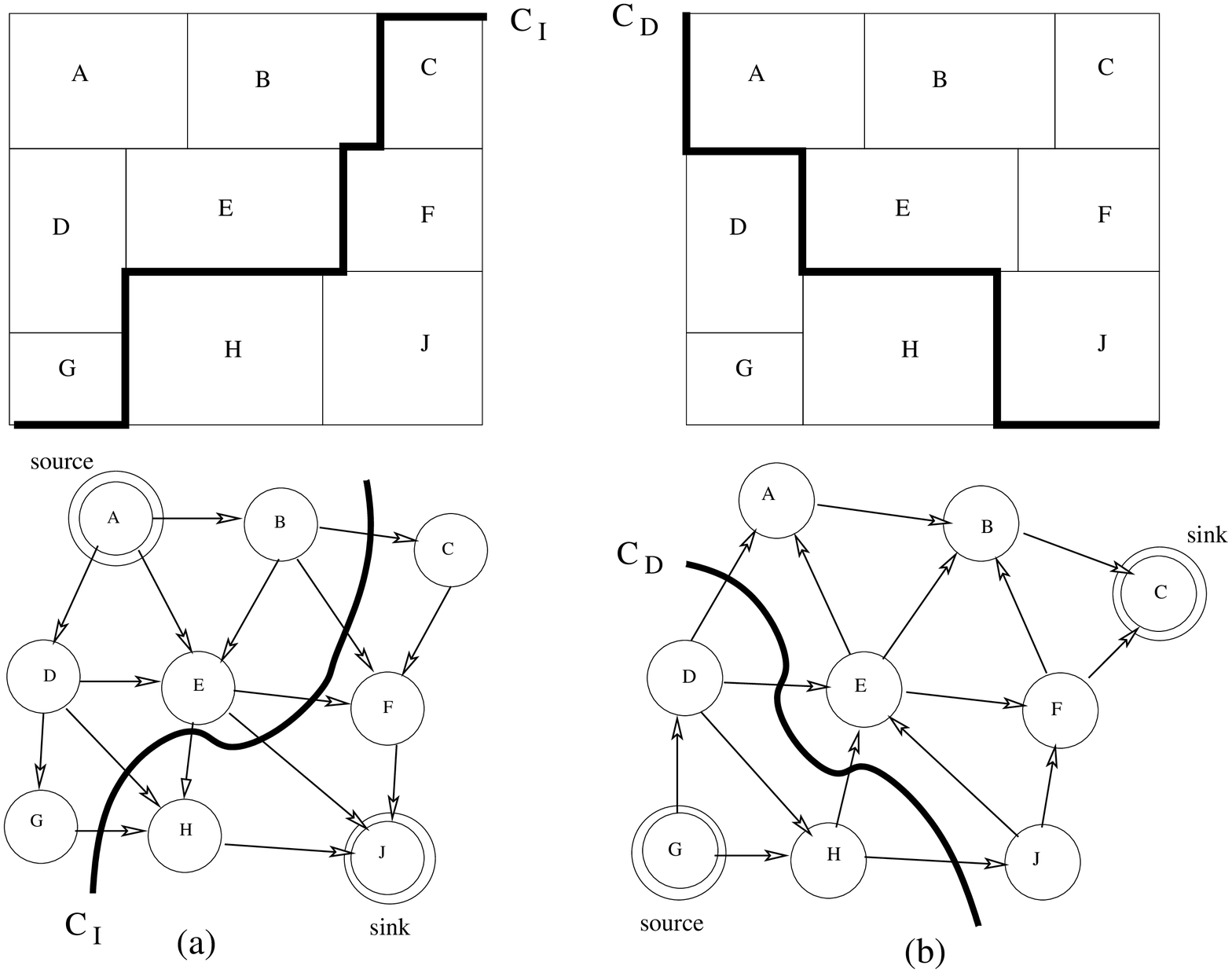}
\caption{A floorplan and the corresponding BAGs for a (a) MIS and (b) MDS cut \cite{karb}}
\label{fig:bag}
\end{figure}

The definition of BAG for obtaining a \textit{monotonically decreasing staircase} (MDS) is as follows: edge $e_{ij}$ =\{($v_i,v_j$) for a pair of adjacent blocks ($b_i, b_j$) such that $b_i$ is to the \textit{left of (below)} $b_j$\}. The source and sink vertices are identified as the vertices pertaining to the bottom-left and top-right corner blocks respectively. This scenario is captured in Fig. \ref{fig:bag} (b) along with the MDS cut $C_D$. In the rest of the paper we refer an MIS/MDS cut as a \textit{ms-cut} unless stated explicitly. It is to be noted that, unlike in \cite{majum1,majum2}, this graph based framework does not consider any netlist information while constructing BAG for faster bipartitioning results. The netlist information is solely used to identify the cut nets and the uncut nets that fall on either side of the bipartition. These uncut nets and the respective parts of the cut nets with atleast two pins in each part. In this method, net cut information in each level of the bipartition hierarchy is a measure of the optimality of each ms-cut obtained, and is referred to as min-cut balanced floorplan bipartitioning \cite{majum1,majum2}.

In order to ensure each cut in BAG is an ms-cut, we refer to the following lemma given in \cite{majum1}, commonly known as \textit{monotone staircase property}.
\begin{lemma} 
\label{lem:1}
If $e_{ij} \in E_b$ is an arc in $G_b$, then there exists at least one monotone staircase in the floorplan such that the blocks $b_i$ and $b_j$ appear in the left and right partitions respectively, and there exists no staircase with $b_i$ in the right partition and $b_j$ in the left partition.
\end{lemma}
\begin{proof}
In \cite{majum1}.
\end{proof}

\begin{figure} [!ht]
\centering
\includegraphics[scale=0.35]{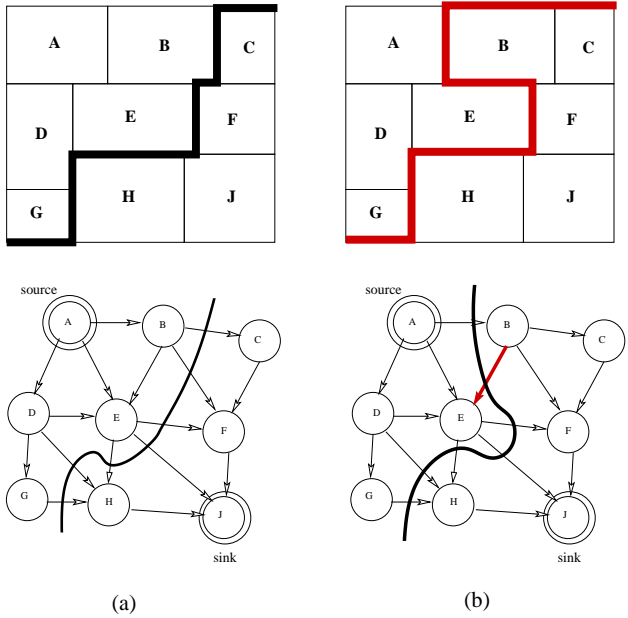}
\caption{Illustration of Lemma \ref{lem:1}: a floorplan with a (a) monotone staircase, and (b) non-monotone staircase}
\label{fig:mscprop}
\end{figure}
In Fig. \ref{fig:mscprop}, we illustrate the working of Lemma \ref{lem:1} for an MIS cut, which is equally applicable for an MDS cut. It shows that all the cut edges in the BAG are forward edges, i.e., directed from the left partition containing the source vertex $A$ towards the right partition containing the sink vertex $J$ yielding a valid monotone staircase cut. However, in Fig. \ref{fig:mscprop} (b), the highlighted edge ($B$,$E$) in the BAG is directed from the right partition to the left partition. This cut leads to a non-monotone staircase cut. From this illustration and Lemma \ref{lem:1}, we observe that it requires at least one back edge directed from the right to left partition to generate a non-monotone staircase cut.

\begin{corollary}
\label{cor:1}
Given a BAG formulated for obtaining a MIS (MDS) cut, any cut which has at least one \textit{back edge} results in a \textit{non-monotone staircase cut}.
\end{corollary}
\begin{proof}
From Lemma \ref{lem:1} and Fig. \ref{fig:mscprop} (b).
\end{proof}

In order to study the advantage of early global routing using monotone staircase patterns over non-monotone staircases, we consider the example in Fig. \ref{fig:mscprop_adv} for two different routing instances of a two pin net $n$ having terminal pins ($A$, $B$). Wirelength for the monotone routing path is equal to half of the bounding box length of the net, i.e., \textit{half perimeter wirelength} (HPWL), while that of the non-monotone path yields extra wirelength beyond HPWL. This eventually consumes more routing area and hence increases the congestion in the routing regions, impacting the routability of the nets. A non-monotone pattern may also require more number of vias depending the number of bends in it. On the other hand, a suitably chosen monotone staircase pattern with fewer bends in it may yield fewer via counts. Therefore, pattern routing using non-monotone staircases is not beneficial for identifying a shortest routing path, as well as fewer via counts. Nevertheless, non-monotone routing \cite{zhang2} or maze routing \cite{leem,sherw} can be effective when monotone or L/Z \cite{kast,zcao} patterns can not be used due to heavy congestion and more routing blockages due to already routed nets within the bounding box of a net. This leads to a detoured routing path with increased wirelength and possibly higher via count, identified using non-monotone or maze routing. 
\begin{figure}[!ht]
\centering
\includegraphics[scale=0.37]{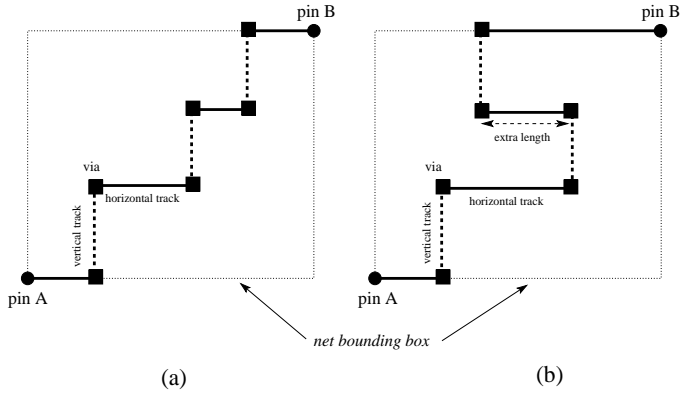}
\caption{Illustrating two routing instances of a $2$-pin net $n$ = ($A$,$B$) using a: (a) monotone staircase path, and (b) non-monotone staircase path}
\label{fig:mscprop_adv}
\end{figure}

Our study also shows that a very large number of monotone routing paths with varying number of bends are possible within the bounding box of a net, L/Z patterns being a subset of all those possible patterns with only one/two via overhead. An optimal monotone pattern is the one which takes minimal number of (bends) vias to complete the routing between a pair of pins through a set of metal layers, thus motivating this work. In this paper, the proposed recursive bipartitioning framework identifies a set of optimal monotone staircases with minimal number of bends in a given floorplan, for early global routing of the nets with minimal wirelength and via count. In this paper, we used only STAIRoute as the early global routing tool, by preferred directional routing in different metal layers.

%% Section: MSC Bipartitioning with Bend Minimization
\section{Monotone Staircase Bipartitioning with Minimal Bends}
\label{sec:bend}
In this section, we discuss the proposed recursive floorplan bipartitioning method in order to identify a set of minimal bend monotone staircase routing regions in a floorplan, for obtaining the shortest routing paths of a set of nets in floorplan, by an early global routing framework such as STAIRoute \cite{karb2}. Before that, we study the impact of a number of bends in a monotone staircase routing region on the number of vias when a net is routed through it, using reserved layer model for layer assignment of the net segments in different routing layers. In this routing model, horizontal and vertical segments of a net are routed through designated metal layers, say $M1$ and $M2$ respectively (see Fig. \ref{fig:mscprop_adv}). This requires inter-layer metal interconnects, called vias, to establish electrical connections between the wire segments of a net running in different layers. For the sake of simplicity, we assume routing with two routing layers ($M1, M2$), although it can be extended to any number of permissible layers in the fabrication processes. 
\begin{figure}[!ht]
\centering
\includegraphics[scale=0.4]{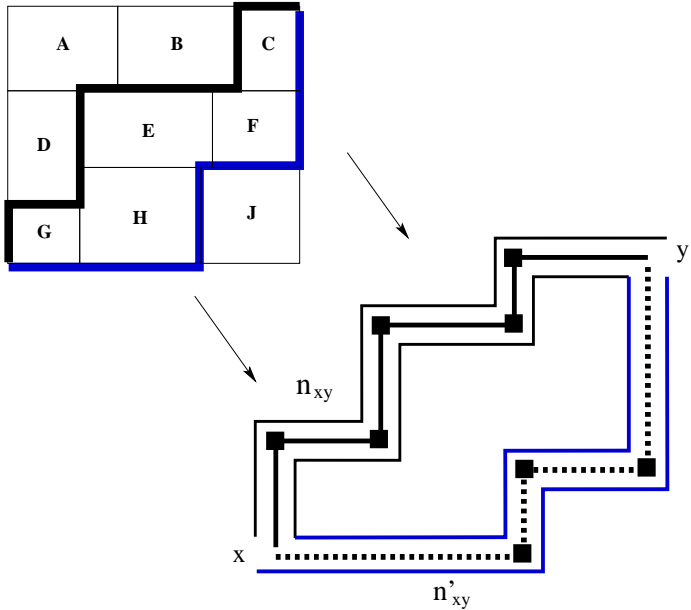}
\caption{Impact of bends in a monotone staircase region on the number of vias (marked as $\blacksquare$)}
\label{fig:bend_via}
\end{figure}

We consider two different routing instances between the terminal points (pins) $x$ and $y$ of a net segment $n$ as depicted in Fig. \ref{fig:bend_via}. These routes, denoted as $n_{xy}$ and $n'_{xy}$ respectively, use different monotone staircase paths with different bend counts. While routing path $n_{xy}$ uses five vias, $n'_{xy}$ requires only three vias. From this example, we infer that a monotone staircase with fewer bends can potentially reduce the number of vias when a net is routed through it using different metal layers and hence serves as the motivation of this work. 

The problem definition in this work is augmented over the existing bipartitioning methods such \cite{dasg,karb,majum1,majum2}, considering a new objective of bend minimization. We enlist the objectives of this new multi-objective optimization problem as below:
\begin{enumerate}
\item balance ratio $balr$ = min($A_l$,$A_r$)/max($A_l$,$A_r$) be maximized
\item the number of cut nets ($k_c$) be minimized, {\it and}
\item the number of bends ($z$) in the monotone staircase be minimized
\end{enumerate}
where  $A_{l(r)} = \sum_{b_i \in B_{l(r)}} \mathrm{Area(b_i)}$, the area of the left (right) partition and $B_{l(r)}$ denotes the set of blocks in the left (right) partition. The number balanced bipartition problem can be seen as a restricted version of the area balanced bipartitioning problem when the area of each block is almost equal, i.e., having negligible variance in block area such that they can be normalized to unity. In this case, $balr$ is defined as min($n_l$,$n_r$)/max($n_l$,$n_r$), where $n_l$ ($n_r$) denotes the number of blocks in the left (right) partition.

A linear combination function of these objectives, with a pair of trade-off parameters ($\gamma$, $\beta$), is defined as below:
\begin{equation}
\label{eq:2}
Gain = \gamma .balr + (1-\gamma -\beta)(1-k_c/k) + \beta(1-z/z_{max})
\end{equation}
where $z_{max}$ is the maximum possible number of bends if the constituent rectilinear segments in the corresponding monotone staircase had alternating (vertical or horizontal) orientation. It is computed as one fewer than the number of segments in it. Notably, Eqn. \ref{eq:2} is similar to that defined in \cite{karb} when $\beta = 0.0$. Careful selection of ($\gamma$, $\beta$) pair may yield an optimal balance among these objectives, not necessarily a global optimum. Since the area balanced bipartitioning is an NP-hard problem \cite{majum2}, the optimum balance among these objectives is hard to obtain in polynomial time. Instead, for a given ($\gamma$, $\beta$) pair, an optimal monotone staircase with maximum $Gain$ is chosen out of those with $Gain$ values in the sequence of $n-1$ bipartitions of a floorplan of $n$ blocks \cite{karb} (see Fig. \ref{fig:bend}), at a given bipartition hierarchy. In Section \ref{sec:result}, we study the bipartitioning results with a range of ($\gamma$, $\beta$) values on a set of floorplan benchmark circuits.

\begin{figure}[!ht]
\centering
\includegraphics[scale=0.36]{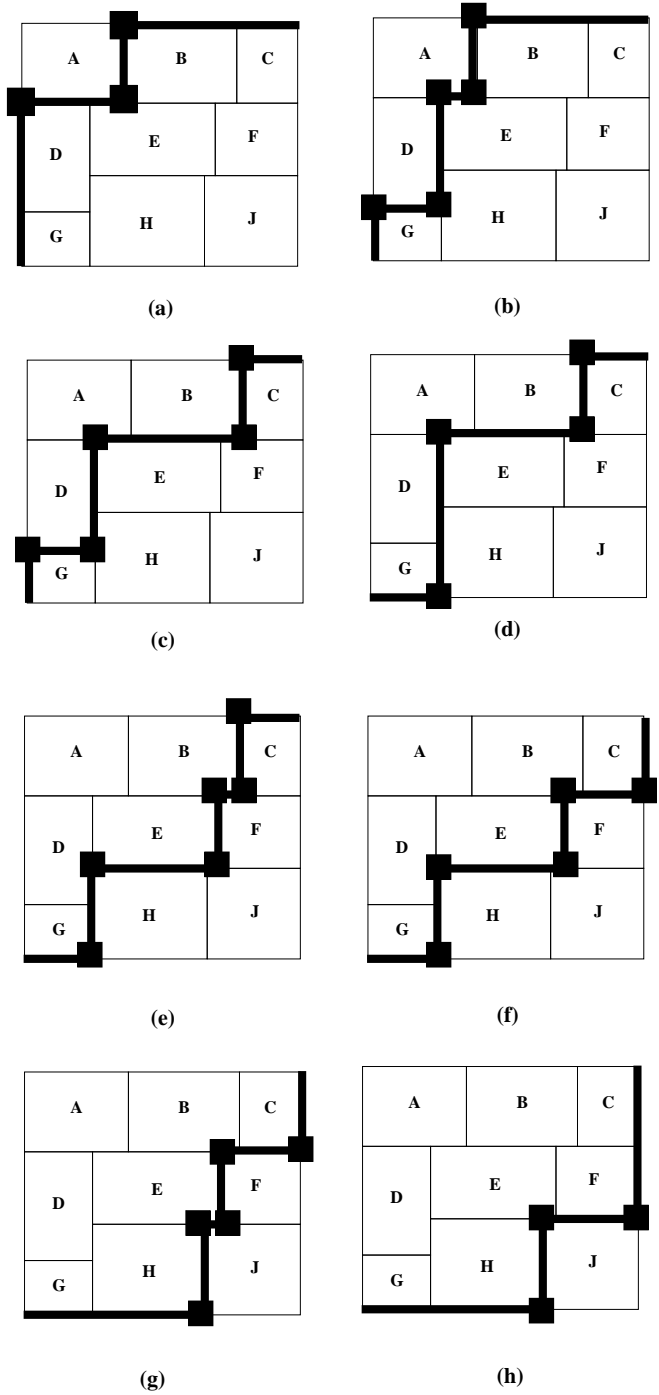}
\caption{A sequence of $n-1$ monotone staircases with varying number of bends (denoted as $\blacksquare$) in a given floorplan: (a) $3$, (b) $5$, (c) $5$, (d) $4$, (e) $6$, (f) $5$, (g) $5$ and (g) $3$}
\label{fig:bend}
\end{figure}

Now we refer to Fig. \ref{fig:bend} for the working of this bipartitioning framework while maximizing the area in each partition and assessing the corresponding bends in the resulting monotone staircase. In this study, we do not consider minimal net cut for the sake of simplicity and restrict only to area balance and minimal bend count. The bipartition instance in Fig. \ref{fig:bend} (a) and (h) gives minimum number of bends ($z=3$), but with poor area balance. The area balance between the partitions keeps on improving through the instances depicted in Fig. \ref{fig:bend} (b)-(e) with varying number of bends, while it declines for instances shown in Fig. \ref{fig:bend} (f)-(h). The best possible area balance may be attained in case of the bipartition in Fig. \ref{fig:bend} (e), but yields the worst bend count ($z=6$) among all others. Therefore, a suitable trade-off between area balance and bend count has to be made based on ($\gamma$, $\beta$) values. The bipartition instance with $z=4$ in Fig. \ref{fig:bend} (d) appears to be a good choice among all the other instances. The following lemma gives a measure of the number of bends in a monotone staircase.
\begin{lemma}
\label{lem:b1}
Given a floorplan with $n$ blocks, the number of bends in a monotone staircase routing region is $O(n)$.
\end{lemma}
\begin{proof}
The number of bends in a monotone staircase can be at most one fewer than the number of cut edges in BAG $G_b$ due to alternate orientation of the contiguous cut edges. Since $G_b$ is a planar graph \cite{karb,majum1} and $|E_b|$ is $O(n)$, the number of cut edges (a subset of $E_b$) that constitutes a monotone staircase is also $O(n)$.
\end{proof}

%In Fig. \ref{fig:bendT}, we illustrate four different orientations of a T-junction and the relation of a T-junction with a bend. Each orientation can either have an unique upward or downward bend belonging to either an MIS or MDS type. Since, the number of T-junctions in any floorplan containing $n$ blocks is $2n-2$ \cite{karb2,ssk}, there are $2n-2$ bends in the entire floorplan. As a results, $8^{2n-2}$ combinations of bends are possible for floorplan topology. A bipartitioning result obtained by our proposed method identifies only $2n-2$ bends for the entire floorplan out of one of these $8^{2n-2}$ combinations. These T-junctions are the basis for constructing the routing graph for early global routing by STAIRoute \cite{karb2}.
%\begin{figure}[!ht]
%\centering
%\includegraphics[scale=0.45]{bend_T_junction.eps}
%\caption{Enumerating different orientations of a T-junction ($\blacksquare$) and the possible bends (in {\color{blue} blue}) pertaining to an MIS/MDS}
%\label{fig:bendT}
%\end{figure}

\subsection{The Algorithm: BFS based Greedy Approach}
The pseudo-code for the proposed monotone staircase bipartitioning method with minimal bends, namely \textit{MSCut\_Bend\_BFS}, is presented in Algorithm \ref{alg:mscutbend}. The inputs to this method are the BAG $G_b$ obtained from a given floorplan $F$ of a set of blocks $B$, a set of nets $N$, the trade-off parameters ($\gamma$, $\beta$) such that $\gamma, \beta \in [0,1]$. The balance type $baltype$ dictates either an \textit{area} or a \textit{number} balanced bipartitioning \cite{majum1,majum2}. Unlike the previous works, we focus on \textit{area} balanced bipartitioning only, since number balanced mode is a special case of it. The key differences between \textit{MSCut\_Bend\_BFS} and the bipartitioning method in \cite{karb} are: (i) bend minimization considered as an additional objective, and (ii) no restriction on the convergence within user-defined area bounds. In rare floorplan instances, these area bounds in \cite{karb} may lead to exploration of a sequence of $n-1$ monotone staircases. On contrary, our method is able to explore a sequence of $n-1$ staircases without any such constraints on any floorplan of $n$ blocks. 
\begin{algorithm}%[!ht]
%\footnotesize
\SetAlgoLined
%\begin{algorithmic}
\SetKwData{Left}{left}\SetKwData{This}{this}\SetKwData{Up}{up}
\SetKwFunction{Union}{Union}\SetKwFunction{FindCompress}{FindCompress}
\SetKwInOut{Input}{input}\SetKwInOut{Output}{output}

\Input{$G_b$, $N$, $\gamma$, $\beta$, $baltype$}
\Output{An optimal monotone staircase for a given ($\gamma$, $\beta$) with maximal area balance, minimal net cut and minimal number of bends}
\BlankLine

Initialize a Queue $Q$ and the left partition $L$ = $\varnothing$\\
Enqueue the source vertex of $G_b$ in $Q$ as (BFS) level $0$ vertex, and include it in $L$ (right partition $R$ = $V_b \setminus L$)\\
/* A vertex once enqueued always remains in $L$ \cite{karb}*/\\
Also enqueue $\varnothing$ as BFS level indicator\\
\While{$Q$ is not empty}{
Let $v_i$ be the dequeued vertex\\
\If{($v_i $ $\neq$ $\varnothing$)}{
\For{($v_j$ $\in$ adj($v_i$))}{
\If{($v_i,v_j$) results in a valid ms-cut (see Lemma \ref{lem:1})}{
Enqueue the vertex $v_j$ and include it in $L$\\
Compute the parameters for the ($L,R$) partition (see Eqn. \ref{eq:2}) and store them in a list $\lambda$
}
}
}
\Else{
Increment BFS level\\
Enqueue $\varnothing$ as next BFS level indicator\\
}
}
Return an optimal monotone staircase with the maximum $Gain$ value $C_{max} \in \lambda$\\
%\end{algorithmic}
\caption{MSCut\_Bend\_BFS}
\label{alg:mscutbend}
\end{algorithm}

The recursive procedure for obtaining a set of minimal bend monotone staircases for the entire floorplan is presented in Algorithm \ref{alg:msctreebend}, by recursively calling \textit{MSCut\_Bend\_BFS} with a set of required inputs. Here, $stype$ dictate the output staircase type, either an MIS or MDS (see Fig. \ref{fig:bag}). In this procedure, the root node of the bipartition hierarchy starts with a particular type e.g. MIS, followed by alternating types in the subsequent levels of the hierarchy. An example of a bipartition (MSC) tree in Fig. \ref{fig:msctree} illustrates a set of optimal monotone staircases (MIS/MDS) with minimal bends are overlaid on an input floorplan of $17$ blocks.
\begin{algorithm}%[ht]
%\footnotesize
\SetAlgoLined
\SetKwData{Left}{left}\SetKwData{This}{this}\SetKwData{Up}{up}
\SetKwFunction{Union}{Union}\SetKwFunction{FindCompress}{FindCompress}
\SetKwInOut{Input}{input}\SetKwInOut{Output}{output}

\Input{$B$, $N$, $F$, $stype$, $\gamma$, $\beta$, $baltype$}
\Output{A bipartition hierarchy (\textit{MSC tree}) with increasing (decreasing) monotone staircases MIS (MDS) at alternate level}
\BlankLine

\If{($Root\_node \parallel TreeLevel \% 2 = 0$)}{
$stype = 1$
}
\Else{
$stype = 0$
}
\BlankLine

$G_b$ = ConstructBAG($B, F, stype$) /* (see Fig. \ref{fig:bag}) */\\
Node.cut = MSCut\_Bend\_BFS($G_b$, $N$, $\gamma$, $\beta$, $baltype$)\\
Node.Level = $TreeLevel$; increment $TreeLevel$\\
\If{($|B_l| \geq 2$)}{
Node.left = Hier\_MSCut\_Bend($B_l$, $N_l$, $F_l$, $stype$, $\gamma$, $\beta$, $baltype$)
}
\BlankLine

\If{($|B_r| \geq 2$)}{
Node.right = Hier\_MSCut\_Bend($B_r$, $N_r$, $F_r$, $stype$, $\gamma$, $\beta$, $baltype$)
}
Return Node.

\caption{Hier\_MSCut\_Bend}
\label{alg:msctreebend}
\end{algorithm}

\begin{figure}[!ht]
\centering
\includegraphics[scale=0.45]{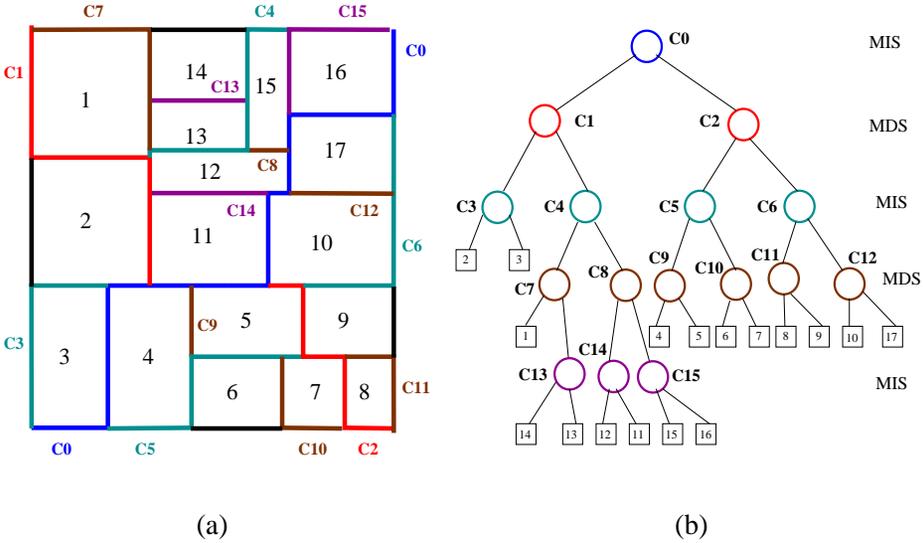}
\caption{A floorplan of $17$ blocks (a) with monotone increasing/decreasing staircases (MIS/MDS), and (b) a (neraly) balanced bipartition tree (MSC tree \cite{karb}) for a ($\gamma,\beta$) pair}
\label{fig:msctree}
\end{figure}

%Theorem
\begin{theorem}
\label{theo:b1}
Given a floorplan with $n$ blocks and $k$ nets, Hier\_MSCut\_Bend takes $O((n^2+nk)\log n)$ time to generate a hierarchy of minimal bend monotone staircases in it.
\end{theorem}
\begin{proof}
Since the block adjacency graph $G_b$ of a given floorplan instance $F$ for $n$ blocks is a planar graph, its construction takes $O(n)$ time. By Lemma \ref{lem:b1}, each while loop in Algoritth \ref{alg:mscutbend} (MSCut\_Bend\_BFS) takes $O(n)$ for identifying $O(n)$ bends and $O(k)$ for net bipartition. Thus, at any recursion level , each call to MSCut\_Bend\_BFS takes $O(n+n^2+nk)$, i.e., $O(n^2+nk)$. Since, MSCut\_Bend\_BFS yields a (nearly) balanced bipartition of the (sub)floorplans at each recursion, the number of levels in the bipartition hierarchy (called MSC tree \cite{karb}) is $O(\log n)$. Therefore, for the entire bipartition hierarchy of $O(\log n)$ levels, the recursive procedure Hier\_MSCut\_Bend takes $O((n^2+nk)\log n)$ time to identify a set of minimal bend monotone staircases for the entire floorplan $F$.
\end{proof}

In Section \ref{sec:result}, we provide a few experimental results to show that the bipartition hierarchy, i.e., MSC tree has $O(\log n)$ height for any floorplan instance $F$ of a circuit containing $n$ blocks with any area distribution.

%% Section: Randomized Algo
\section{A New Randomized Neighbor Search Approach}
\label{sec:msc-rand}
\begin{figure*}[!ht]
\centering
\includegraphics[scale=0.3,angle=90]{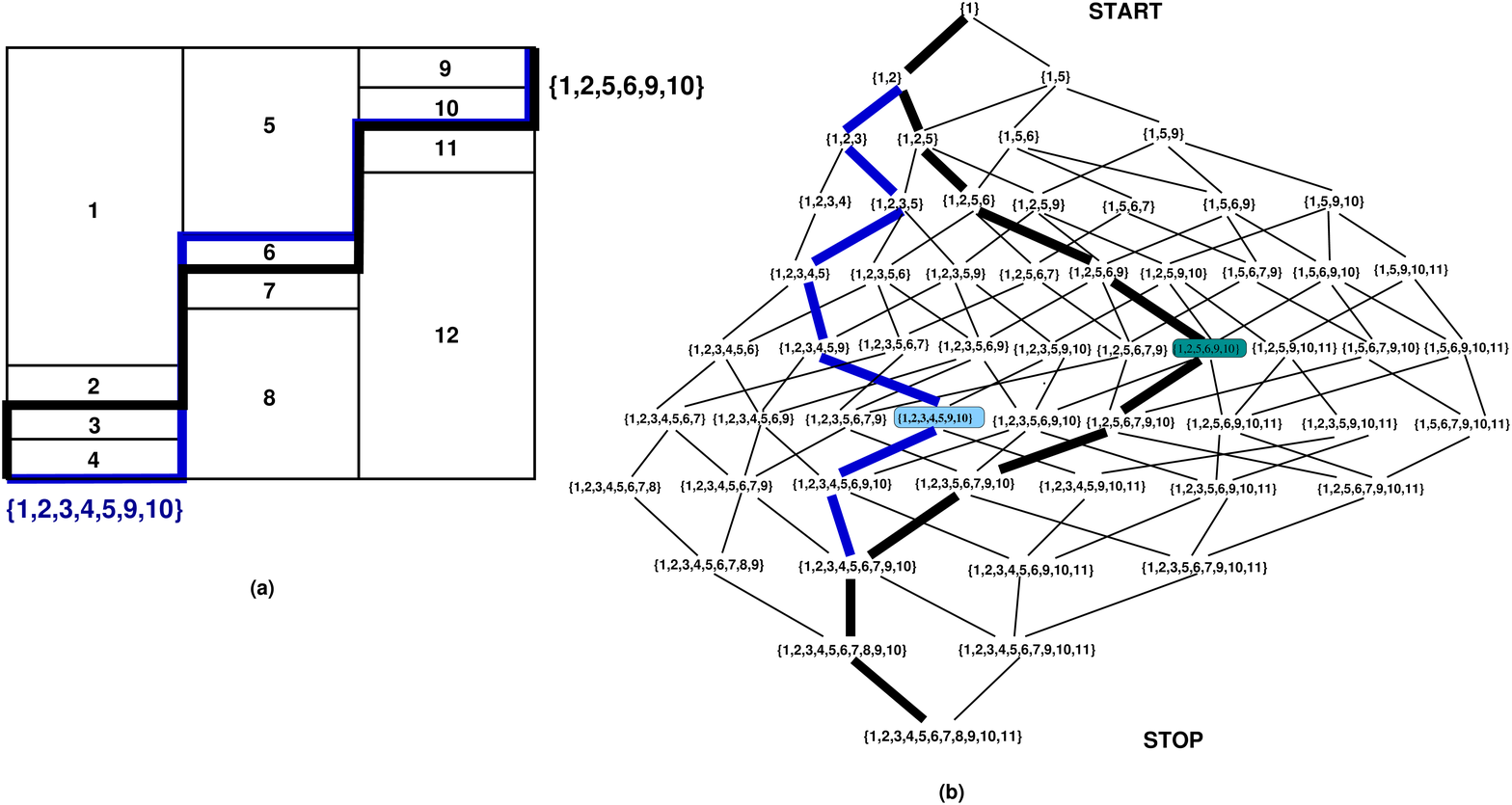}
\caption{A floorplan of $12$ blocks (a) overlaid with an optimal monotone staircase \textbf{\color{blue} $\{1,2,3,4,5,9,10\}$} and a near optimal staircase $\{1,2,5,6,9,10\}$, and (b) the corresponding hasse diagram \cite{hasse} for exponentially large number of sequences (paths) of monotone staircases: the \textbf{\color{blue} blue} path containing \textbf{\color{blue} $\{1,2,3,4,5,9,10\}$} and \textbf{black} path containing $\{1,2,5,6,9,10\}$}
\label{fig:hasse}
\end{figure*}

Given floorplan $F$ for a set of $n$ blocks, the number of all possible monotone staircases in $F$ is exponentially large. Hence, the problem of finding the optimum monotone staircase is known to be NP-Hard \cite{dasg,majum2}. As discussed in Section \ref{sec:bend}, an (near) optimal solution of monotone staircase bipartition implies a suitable trade-off between the constituent objectives: (a) maximizing the area of each bipartition, (b) minimizing the number of bends in the corresponding monotone staircase, and (c) the number of cut nets by this bipartition. Since the area balanced bipartitioning is an NP-Hard problem \cite{majum2}, no polynomial time algorithm exists. Hence, several greedy heuristic approaches have been proposed in \cite{dasg,karb,majum1,majum2} and in Section \ref{sec:bend}. In all cases, a monotone staircase cut with the maximum $Gain$ value pertaining to a given trade-off among the objectives is considered as an optimal bipartition. As stated in Section \ref{sec:bend}, we pick an (nearly) optimal monotone staircase among a sequence of $n-1$ monotone staircases for a given ($\gamma, \beta$) pair. Intuitively, different ($\gamma, \beta$) pairs may potentially yield different optimal solution(s) and even a different sequence.

Given a set $B$ of $n$ blocks for a given floorplan $F$, a monotone staircase bipartition ($L$, $R$) represents a proper subset of $B$. In other words, the blocks in the left partition $L$ (hence the right partition $R$ = $B\setminus L$) constitute a proper subset of $B$, while obeying the monotone staircase property (refer to Lemma \ref{lem:1} \cite{majum1}). Thus, ($L$, $R$) represents a valid monotone staircase cut on the block adjacencyy graph (BAG) for $F$. In summary, the set of all possible monotone staircases $S_m$ in $F$ is a subset of power set of $B$ ($S_m \subseteq power(B)$). Notably, $S_m$ is a partially ordered set by inclusion $\subseteq$ operation on all the monotone staircases in $F$ that can be identified in exponential time. A staircase $s_m \in S_m$ covers a set of one or more staircases $\{s_n\} \subseteq S_m$ if $s_m$ can be obtained from $\{s_n\}$. Based on this, we construct the corresponding \textit{hasse diagram} \cite{hasse} pertaining to $S_m$. An example hasse diagram for a floorplan of $n$ = $12$ (and $|S_m| = 56$) is illustrated in Fig. \ref{fig:hasse}. 

In Section \ref{sec:bend} and also in \cite{karb}, we studied that a sequence of $n-1$ monotone staircases can be identified greedily at any level of bipartition hierarchy by the respective bipartitioning methods. An optimal solution is identified from this sequence based on a given trade-off ($\gamma, \beta$). There is a scope of obtaining an improved solution if more than $n-1$ staircases can be explored, with proportionally higher runtime overhead. In this section, we present a new technique for exploring the neighbors of a block (vertex) in the BAG $G_b$ for a potentially better optimal monotone staircase (obeying Lemma \ref{lem:1} in terms of the objectives considered. We study how selection of a neighbor is done based on random indexing of the neighbors of vertex $v_i$ in $G_b$. Alike the BFS based method (see Algorithm \ref{alg:mscutbend}), the proposed bipartitioning method also adopts BFS on $G_b$. However, this method can more aptly resemble with an wave-front propagation in Ether. This may lead to different sequences of (not necessarily disjoint) monotone staircase cuts on the BAG. In Fig. \ref{fig:hasse} (b), an example of these sequences are highlighted by different paths, one with black and other by blue color, from \textit{START} to \textit{STOP} node in the \textit{hasse diagram}. In this diagram, each node represent a distinct monotone staircase and edges represent their possible transition to another distinct monotone staircase. In other words, these edges represent the inclusion operation. While the directed search method in Section \ref{sec:bend} identifies only one sequence marked by the bold black line in Fig. \ref{fig:hasse} (b), the randomized method under discussion identifies different sequences during different trials of the proposed randomized neiighbor search technique. Notably, the number of sequences obtained by the random method can be more than one, but are not necessarily maximally disjoint. It is also evident that the length of such a path (\textit{START} $\rightsquigarrow$ \textit{STOP}) is always $n-1$ as stated in Lemma of \cite{karb}. However, the number of such paths grow exponentially with $n$ and the sequences (hence the Hasse diagram) also differ due to different floorplan topology for the same set of blocks $B$. A comparative study of staircase wave-front propagation using greedy and randomized neighbor search technique is presented in Appendix \ref{apdx:1}.

In Fig. \ref{fig:hasse} (b), we consider an example of two different sequences of $n-1$ monotone staircases marked by the \textbf{\color{blue} blue} and \textbf{black} lines, out of exponentially large number of possible sequences between \textit{START} and \textit{STOP} nodes. Here START and STOP nodes denote trivial monotone staircases containing only one block in the left (right) partition. If one path does not contain an optimal monotone staircase, another path may be explored in a hope to identify an optimal one. Since area balanced monotone staircase bipartitioning is a NP-hard problem, there is no method that verifies such a scenario, unless we apply the brute force method to explore all possible sequences. However, a random transition from one node to another may lead to traversing a new path either completely or partially. Careful study of Fig. \ref{fig:hasse} (b) shows that randomization at the suitable node, say in this case $\{1,2\}$, choosing the block $3$ randomly instead of $5$ (by greedy approach) may guide to a different sequence leading to a potentially optimal solution $S_{opt} = \{1,2,3,4,5,9,10\}$, for a given ($\gamma$, $\beta$) pair. In summary, several such randomized selections (while traversing from \textit{START} $\rightsquigarrow$ \textit{STOP}) may yield an optimal solution or a scope of obtaining a better solution than the previously found optimal solution. A number of such trials may be exercised in order to explore partially/completely different sequences and thus obtain a potentially better solution. However, in order to contain the run time within the same bound as in Section \ref{sec:bend}, a large number of such trials can not be afforded. Instead, we restrict the number of trials to a reasonably small number and use random seeds for each trial. After all such trials, an optimal monotone staircase is identified as the one, with maximum $Gain$ value, among all the staircases explored along different paths in \textit{START} $\rightsquigarrow$ \textit{STOP}.
\begin{figure}[!ht]
\centering
\includegraphics[scale=0.4]{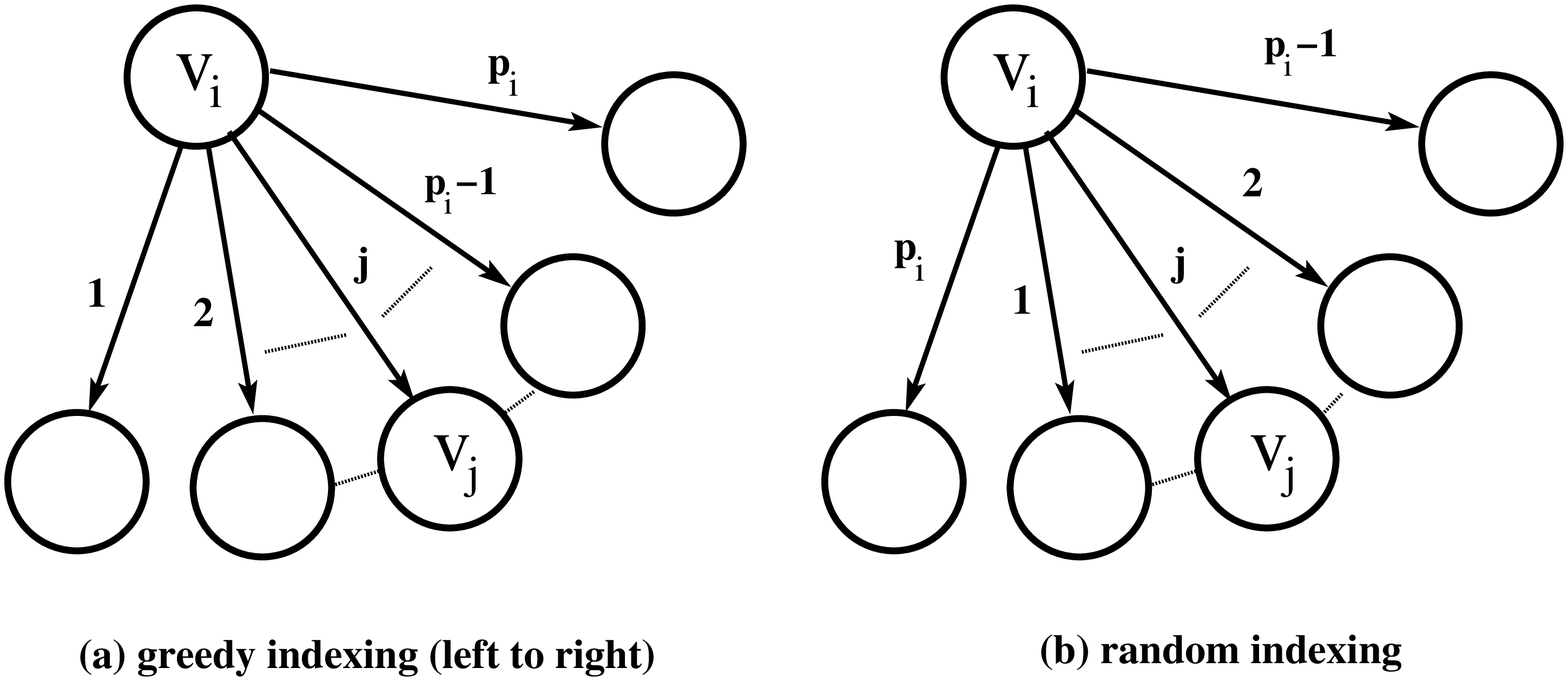}
\caption{Exploring the neighbors of $v_i$ in $G_b$ based on their indexing}
\label{fig:neigh}
\end{figure}

In the proposed randomized neigbor search method \cite{karb4}, the underlying process of randomly indexing the neighbors of a vertex $v_i$ of $G_b$ brings in the difference with the greedy methods \cite{karb,karb3}. Unlike greedy indexing approach, from left to right used in \cite{karb} and also in Section \ref{sec:bend} (see Fig. \ref{fig:neigh} (a)), a neighbor $v_j$ of $v_i$ with out-degree $p$ is indexed with randomly chosen number $j \in [1,p]$ (see Fig. \ref{fig:neigh} (b)). As in \cite{karb} (also Section \ref{sec:bend}), identifying a set of monotone staircases while exploring all adjacent vertices of $v_i$ takes $p$ time for exploring all the neighbors. The following lemma shows that the average runtime improves.
\begin{lemma}
\label{lem:r1}
For a given vertex $v_i$ with out-degree $p$ in $G_b$, the expected time $E[t_{adj}]$ to search its adjacent list to identify one or more distinct monotone staircases is $(p+1)/2$.
\end{lemma}
\begin{proof}
Since all the vertices in the neighborhood of $v_i$ are equally probable to be picked, with a probability of $1/p$, the expected runtime to search a particular neighbor $v_j$ with random indexing $j \in [1,p]$ is:\\
\qquad $E[t_{adj}]$ = $\sum\limits_{j=1}^{p} (1/p).j $\\
\qquad = $(p+1)/2$\\
\end{proof}

Alike the greedy method in \cite{karb} and Section \ref{sec:bend}, the best case scenario occurs when all the $p$ edges emanating from $v_i$ obey the monotone staircase property (Lemma \ref{lem:1}), thus giving $p$ distinct monotone staircases. The worst case scenario occurs when the number of such edges is only $1$, resulting in only one monotone staircase. The following lemma gives the average number of staircases can be explored by a single vertex $v_i$.
\begin{lemma}
\label{lem:r2}
For a given vertex $v_i$ with out-degree $p$ in $G_b$, $O(p)$ distinct monotone staircases can be identified while obeying Lemma \ref{lem:1}.
\end{lemma}
\begin{proof}
Since, all the $p$ edges emanating from $v_i$ have $1/2$ probability of obeying Lemma \ref{lem:1}, the average case\\
\qquad = $1/p(1 + 2 + ....... + (p-1) + p)$\\
\qquad = $(p+1)/2$\\
Hence, $O(p)$ distinct monotone staircases can be identified.
\end{proof}

\subsection{The Pseudo-code for the proposed randomized bipartitioner}
In this section, we present the pseudo-code for the proposed randomized floorplan bipartitioning method \textit{MSCut\_Bend\_RAND} in Algorithm \ref{alg:mscutrand}, in order to identify a minimal bend monotone staircase in a given floorplan at a given level of bipartition hierarchy. Alike Algorithm \ref{alg:mscutbend}, this algorithm is called at any level of the bipartitioning hierarchy. The bipartition hierarchy is obtained by the same recursive framework presented in Algorithm \ref{alg:msctreebend}.
\begin{algorithm}%[ht]
%\footnotesize
\SetAlgoLined
%\begin{algorithmic}
\SetKwData{Left}{left}\SetKwData{This}{this}\SetKwData{Up}{up}
\SetKwFunction{Union}{Union}\SetKwFunction{FindCompress}{FindCompress}
\SetKwInOut{Input}{input}\SetKwInOut{Output}{output}

\Input{$G_b$, $N$, $\gamma$, $\beta$, $baltype$}
\Output{An optimal monotone staircase for a given ($\gamma$, $\beta$) with maximal area balance, minimal net cut and minimal number of bends}
\BlankLine

Define a Queue $Q$, a list $\lambda$ and iterator $r$ = $0$\\
\While{($r < 3$)}{
Initialize left partition $L$ = $\varnothing$ (right partition $R$ = $V_b \setminus L$)\\
Enqueue the source vertex of $G_b$ in $Q$ as (BFS) level $0$ vertex, and include it in $L$ \\
Also enqueue $\varnothing$ as BFS level indicator\\
\While{(NOT\_EMPTY($Q$))}{
Let $v_i$ be the dequeued vertex\\
Define a wavefront $V_{list} = \varnothing$\\
\If{($v_i $ $\neq$ $\varnothing$)}{
$V_{list} \leftarrow \{v_i\}$\\
}
\Else{ 
\While{There exists at least one cut edge in $E_b$ emanating the vertex front $V_{list}$ and terminating on $R$}{
Generate a random seed to choose a cut edge ($v_i,v_j$), such that $v_i \in V_{list}$ and $v_j \in R$\\
\If{($v_i,v_j$) yields a valid ms-cut (see Lemma \ref{lem:1})}{
Enqueue the vertex $v_j$ and include it in $L$\\ %Enqueue($v_j$)\\
Mark the edge ($v_i,v_j$) as explored\\
Compute the parameters for the ($L,R$) partition (see Eqn. \ref{eq:2}) and store them in a list $\lambda$
}
}
Increment BFS level\\
Enqueue $\varnothing$ as next BFS level indicator\\
}
}
Increment $r$\\
}
Return optimal monotone staircase with maximum $Gain$ $C_{max} \in \lambda$
%\end{algorithmic}
\caption{MSCut\_Bend\_RAND}
\label{alg:mscutrand}
\end{algorithm}

\begin{lemma}
\label{lem:r3}
The proposed randomized bipartitioning method MSCut\_Bend\_RAND takes $O(n^2+nk)$ time for obtaining an optimal monotone staircase with minimal bend count on BAG of a given floorplan $F$.
\end{lemma}
\begin{proof}
Since the number of edges $|E_b|$ in $G_b$ is $O(n)$ and $E_b$ = $\sum\limits_{i=1}^{n} \mathrm{p_i}$, where $p_i$ is the out-degree of $v_i$, it takes $O(n)$ time for searching distinct monotone staircases. In this method, we use $3$ trials in order to obtain a different sequence of monotone staircases in each trial, but possibly not disjoint. Also the net partitioning procedure takes $O(k)$, while finding the number of bends account for $O(n)$ time (see Lemma \ref{lem:b1}). Thus, the overall time taken by the proposed bipartitioning method is $O(n(n+k))$, i.e., $O(n^2+nk)$.
\end{proof}

Note that Algorithm \ref{alg:mscutrand} has the same $O(n^2+nk)$ time complexity as Algorithm \ref{alg:mscutbend}, but only a constant times higher due to multiple trials conducted for obtaining different sequences. In order to obtain a set of optimal monotone staircases with minimal bend count for the entire floorplan, the same recursive bipartitioning framework presented in Algorithm \ref{alg:msctreebend} can be used. Therefore, the recursive procedure considering the proposed randomized technique takes $O((n^2+nk)\log n)$ time to generate a hierarchy of monotone staircase cuts for a given floorplan topology.

\section{Experimental Results}
\label{sec:result}
In order to verify the correctness and efficiency the proposed bipartitioning methods, we ran them on MCNC/GSRC floorplanning benchmark circuits \cite{parque} (see Table \ref{tab:bench}). Different floorplan instances of a circuit were generated using \textit{Parquet} floorplacement tool \cite{adya,parque} using random seeds. In order to observe different bipartitioning scenarios for the same circuit, we generated four different floorplan instances for each circuit. The algorithms were implemented in $C$ programming language and run on a Linux platform ($2.8$GHz, $16$GB RAM).  
%%Table1: Benchmark Details
\begin{table}[!ht]
\footnotesize
\centering
\caption{Floorplanning Benchmarks \cite{parque}}
\begin{tabular}{ |c|c|r|r|r|}
\hline
{Suite} & {Circuit} & {\#Blocks} & {\#Nets}  &  {Avg. Net}\\
& & & & Degree \\  \hline
{MCNC} & apte & 9 & 44 & 3.500\\
\hline
& hp & 11 & 44 & 3.545\\
\hline
& xerox & 10  & 183 & 2.508\\
\hline
& ami33 & 33 & 84 & 4.154\\
\hline
& ami49 & 49 & 377 & 2.337\\ \hline
\hline
{GSRC} & n10 & 10 & 54 & 2.129\\
\hline
& n30 & 30 & 147 & 2.102\\
\hline
& n50 & 50 & 320 & 2.112\\
\hline
& n100 & 100 & 576 & 2.135\\
\hline
& n200 & 200 & 1274 & 2.138\\
\hline
& n300 & 300 & 1632 & 2.161\\ \hline
\end{tabular}
\label{tab:bench}
\end{table}

\subsection{Bipartitioning Results}
In our experimental setup, we ran the proposed monotone staircase bipartitioning methods with minimal bends, BFS (see Algorithm \ref{alg:mscutbend}) and randomized (RAND) version (refer to Algorithm \ref{alg:mscutrand}) that works in breadth-first traversal (BFS) fashion at any node of the bipartition hierarchy (see Algorithm \ref{alg:msctreebend}). For experimental purpose, we also came up with a variant of the BFS based greedy method by adopting depth-first search (DFS) on the BAG. Due to lack of space, we are unable to present its pseudo-code. An example showing the working of these bipartitioning methods (BFS, DFS, RAND) is presented in Appendix \ref{apdx:1}. 

These experiments were conducted with $\gamma \in [0.1$, $0.7]$ and $\beta \in [0.0$,$0.3]$, both varying in steps of $0.1$ such that $\gamma + \beta <= 1$. The corresponding bipartitioning results for BFS, DFS, and RAND methods are presented in Fig. \ref{fig:msc-cost} for: (a) area balance ratio ($balr$), (b) normalized bend count ($z/z_{max}$), (c) normalized net cut ($k/k_c$), and (d) $Gain$ (see Eqn. \ref{eq:2}) respectively. The corresponding values were computed as an average of the respective parameters over the specified ($\gamma, \beta$) pairs and all $4$ instances of a given circuit. We compare these results with an earlier BFS based directed search method \cite{karb} which did not consider bend minimization (BFS-NB). It is also important to note that the results presented in \cite{karb} is for $\gamma$ = $0.4$ only which is similar to the results for $\gamma$ = $0.4$ and $\beta$ = $0.0$ case in BFS mode. Moreover, they did not report the individual objective values in their paper. For fair comparison, we ran their code \cite{karb} for obtaining the results for each of the objectives other than $Gain$ in BFS-NB mode, including runtime.
\begin{figure*}[!ht]
\centering
\begin{subfigure}[b]{0.48\textwidth}
\centering
\includegraphics[scale=0.35]{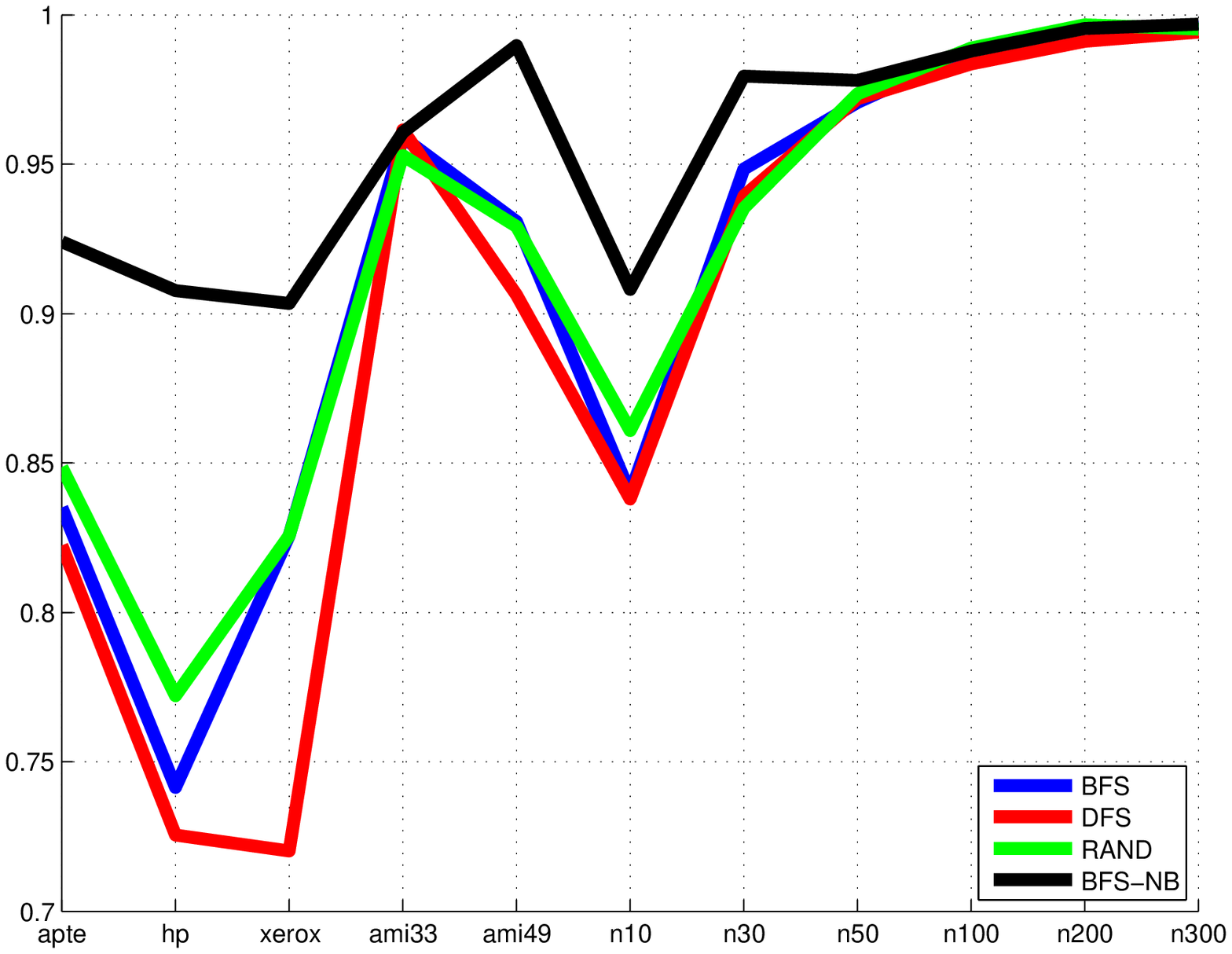}
\caption{Area Balance Ratio ($balr$)}
\end{subfigure}
\begin{subfigure}[b]{0.48\textwidth}
\centering
\includegraphics[scale=0.35]{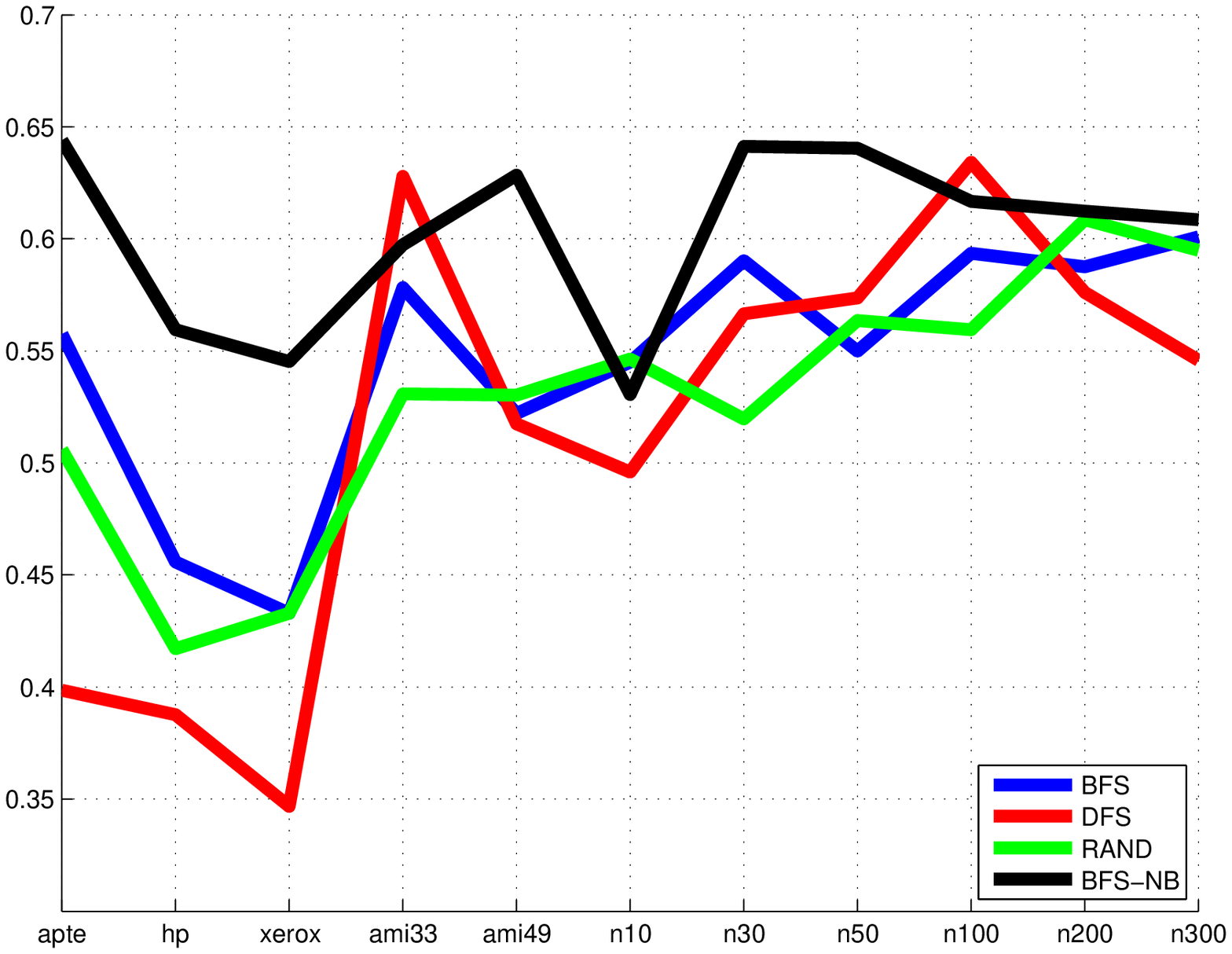}
\caption{Bend Ratio ($z/z_{max}$)}
\end{subfigure}
\begin{subfigure}[b]{0.48\textwidth}
\centering
\includegraphics[scale=0.35]{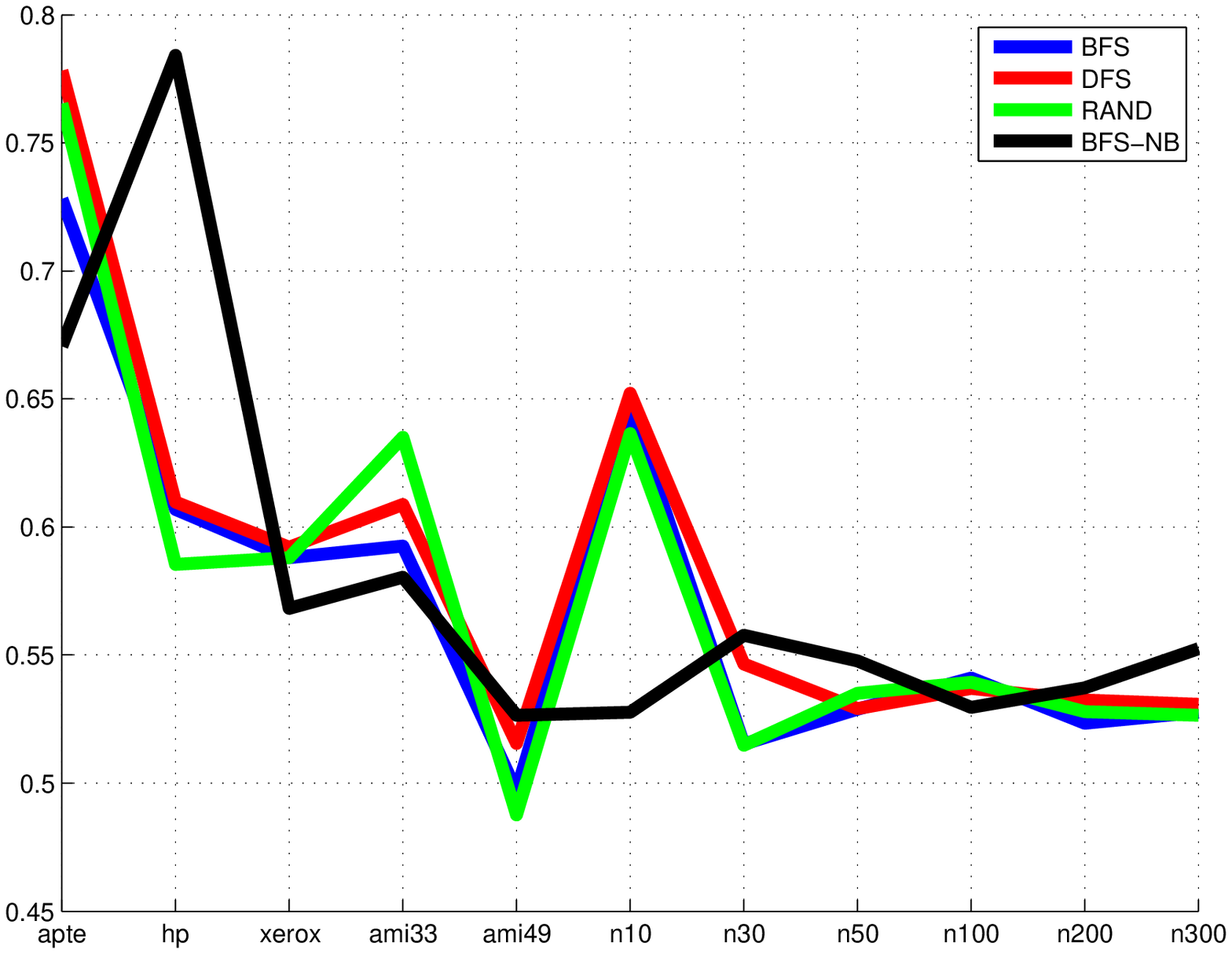}
\caption{Net Cut Ratio ($k_c/k$)}
\end{subfigure}
\begin{subfigure}[b]{0.48\textwidth}
\centering
\includegraphics[scale=0.35]{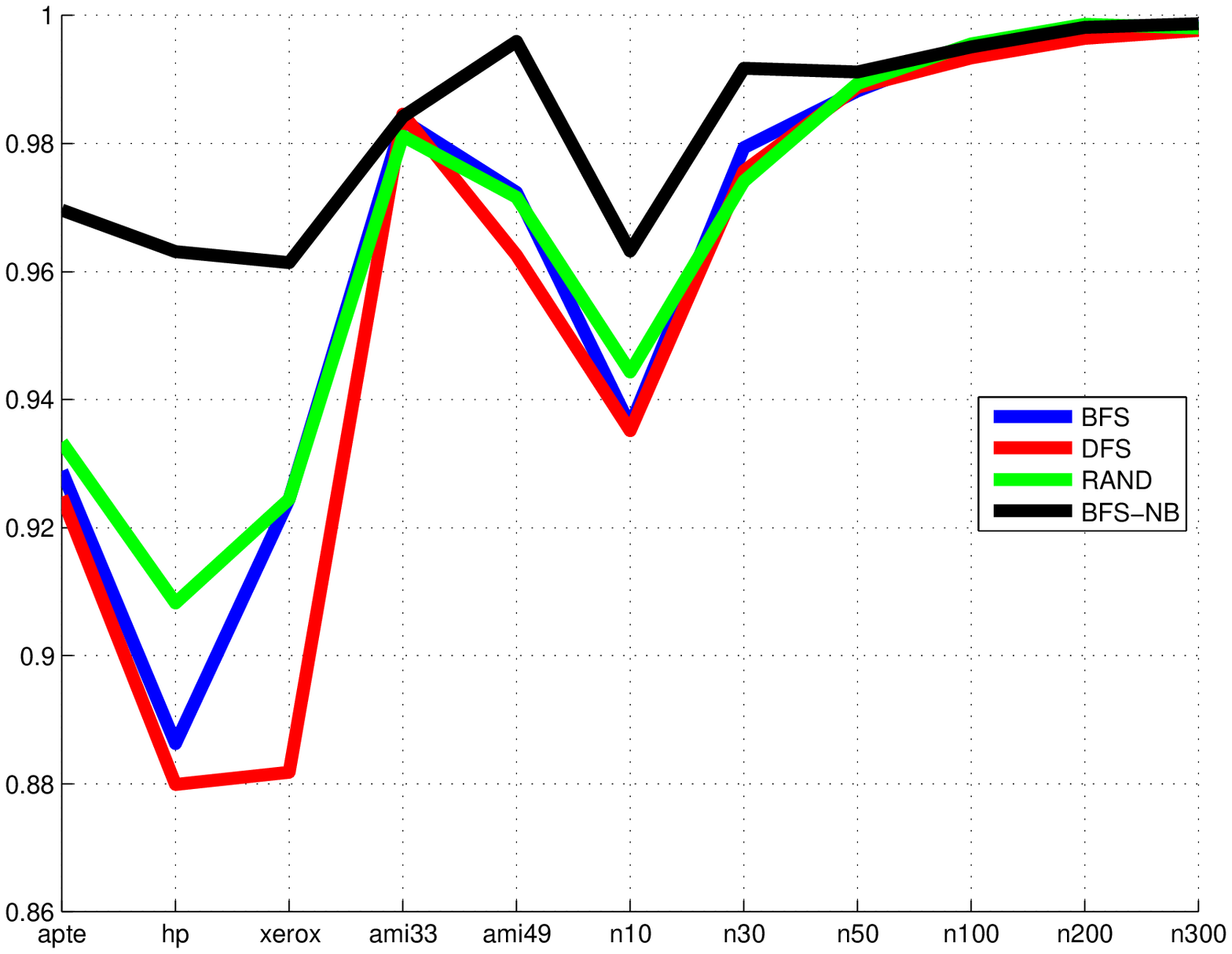}
\caption{Maximum $Gain$}
\end{subfigure}
\caption{Comparison of bipartitioning results: this work (BFS, DFS, RAND) vs BFS-NB (available for $\gamma$ = $0.4$) \cite{karb}}
\label{fig:msc-cost}
\end{figure*}

The results on area balance in Fig. \ref{fig:msc-cost} (a) show that BFS-NB \cite{karb} outperforms all other modes \{BFS,DFS,RAND\} that used bend minimization objective, by focusing on area balance and net cut only. Among the proposed methods, DFS has the worst area balance values for most of the circuits. For net cut, BFS-NB mode performs well only for a few circuits although the net cut objective has more weight of $0.6$ for $\gamma$ = $0.4$. BFS and RAND have better net cut results for most of the circuits. Likewise, DFS mode continues to give higher net cut values for all the circuits. Regarding the number of bends, RAND mode is consistently better for most of the circuits compared to BFS and DFS. Due to certain floorplan topologies in specific circuits, DFS mode had better average values of bend counts for smaller circuits such as $apte$, $hp$, $xerox$, $n10$ with around $10$ blocks and large circuit $n300$. Lastly, BFS-NB consistently yielded the worst (highest) bend counts over other modes. Overall, the $Gain$ values reported for each circuit show that BFS-NB is the best for circuits up to $n50$, followed by RAND mode which dominates the $Gain$ values over BFS and DFS modes for the remaining circuits. For larger circuits like $n50$ and above, RAND mode is seen to supersede BFS-NB with the maximum $Gain$ values. 

Due to balanced bipartitioning at each node of the bipartition hierarchy (MSC tree \cite{karb}), the height of the bipartition (MSC) tree is stated to be $O(\log n)$, where $n$ is the number of blocks in a floorplan. The results presented in Fig. \ref{fig:msc-height} for each circuit shows that the average height of the MSC tree taken over the generated floorplan instances and ($\gamma$, $\beta$) values, is contained within the tight bounds of $\log n$ and $2\log n$, thus establishing the claim in Theorem \ref{theo:b1}.
\begin{figure}[!ht]
\centering
\includegraphics[width=3.2in,height=1.75in]{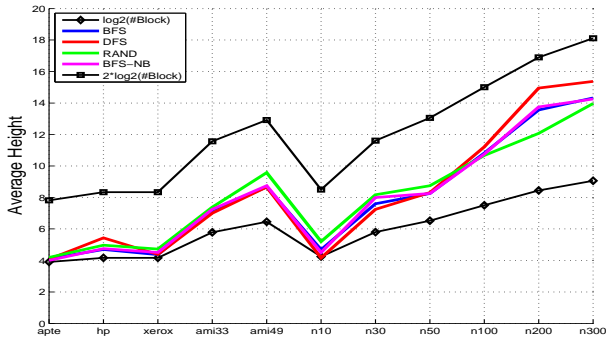}
\caption{Experimental results on the height of MSC tree}
\label{fig:msc-height}
\end{figure}

Table \ref{tab:runtime} presents the runtime results for the proposed recursive floorplan bipartitioners (BFS, RAND and DFS) as well as \cite{karb} (BFS-NB). As stated in Section \ref{sec:msc-rand}, RAND mode is merely a constant times higher than the other two modes and is more prominent with larger circuits such as $n100$, $n200$ and $n300$, while BFS/DFS report similar runtime for all the circuits. But, none of these methods can match the runtime values obtained by the faster method BFS-NB as claimed by \cite{karb} even for the larger circuits.
\begin{table}[!ht]
  \centering
  \caption{Comparison of CPU time ($sec$)}
    \begin{tabular}{|c|r|r|r|r|}
      \hline
      \textbf{Circuit} & \textbf{BFS} & \textbf{DFS} & \textbf{RAND} & \textbf{BFS-NB} \cite{karb} \\
      \hline
      apte & 0.005 & 0.005 & 0.007 & 0.005 \\ \hline
      hp & 0.006 & 0.005 & 0.009 & 0.003 \\ \hline
      xerox & 0.011 & 0.011 & 0.016 & 0.009 \\ \hline
      ami33 & 0.033 & 0.032 & 0.060 & 0.014 \\ \hline
      ami49 & 0.107 & 0.106 & 0.225 & 0.023 \\ \hline
      n10 & 0.005 & 0.004 & 0.008 & 0.008 \\ \hline
      n30 & 0.031 & 0.031 & 0.058 & 0.006 \\ \hline
      n50 & 0.124 & 0.123 & 0.220 & 0.050 \\ \hline
      n100 & 0.803 & 0.800 & 1.369 & 0.062 \\ \hline
      n200 & 7.841 & 7.833 & 12.612 & 0.432 \\ \hline
      n300 & 21.945 & 22.055 & 38.967 & 0.656 \\ \hline
      \hline
      Normalized & & & & \\ 
      Geo Mean & 3.601 & 3.549 & 6.200 & 1.000 \\
      \hline
    \end{tabular}
  \label{tab:runtime}
\end{table}

\subsection{Via Count in Early Global Routing by STAIRoute}
In this section, we present the experimental results on early via estimation by performing early global routing of the corresponding floorplan level netlist using STAIRoute \cite{karb2} and the bipartitioning results presented in earlier subsection for BFS, DFS, and RAND modes. A maximum of $8$ metal layers were used by STAIRoute using preferred routing directions. We present the corresponding results for the largest benchmark circuit $n300$ in Fig. \ref{fig:gr-via-n300-bc} and \ref{fig:gr-via-n300-wc} for $\beta \in \{0.0, 0.1, 0.2, 0.3\}$ and $\gamma \in [0.1, 0.7]$ in steps of $0.1$. This experimental setup does not apply to BFS-NB mode since the corresponding values of ($\gamma, \beta$) is not applicable for it. However, our study confirmed that the via count for BFS-NB mode resembles that with BFS mode for $\gamma$ = $0.4$ and $\beta$ = $0.0$.

We also study the variation of via count for two different floorplan instances of $n300$, the best-case instance with smaller HPWL (Instance\#1) and the worst-case instance with larger HPWL (Instance\#2) in Fig. \ref{fig:gr-via-n300-bc} and \ref{fig:gr-via-n300-wc}. In case of instance\#1, DFS mode dominates over BFS and RAND modes only for $\beta$ = $0.0$. However, $\beta > 0$ cases show that RAND mode dominates DFS for upto some $\gamma$ values, such as $0.3$, $0.5$ and $0.4$ respectively, for the respective $\beta \in \{0.1, 0.2, 0.3\}$. Beyond these $\gamma$ values, DFS yields the best via count for this floorplan instance of $n300$. In a very small range of $\gamma$ and $\beta$ values, i.e., $\beta$ = $0.3$ and $0.4 < \gamma \leq 0.5$, BFS appears to dominate over DFS and RAND modes. 

\begin{figure}[!ht]
\centering
\begin{subfigure}[b]{0.48\textwidth}
\centering
\includegraphics[scale=0.35]{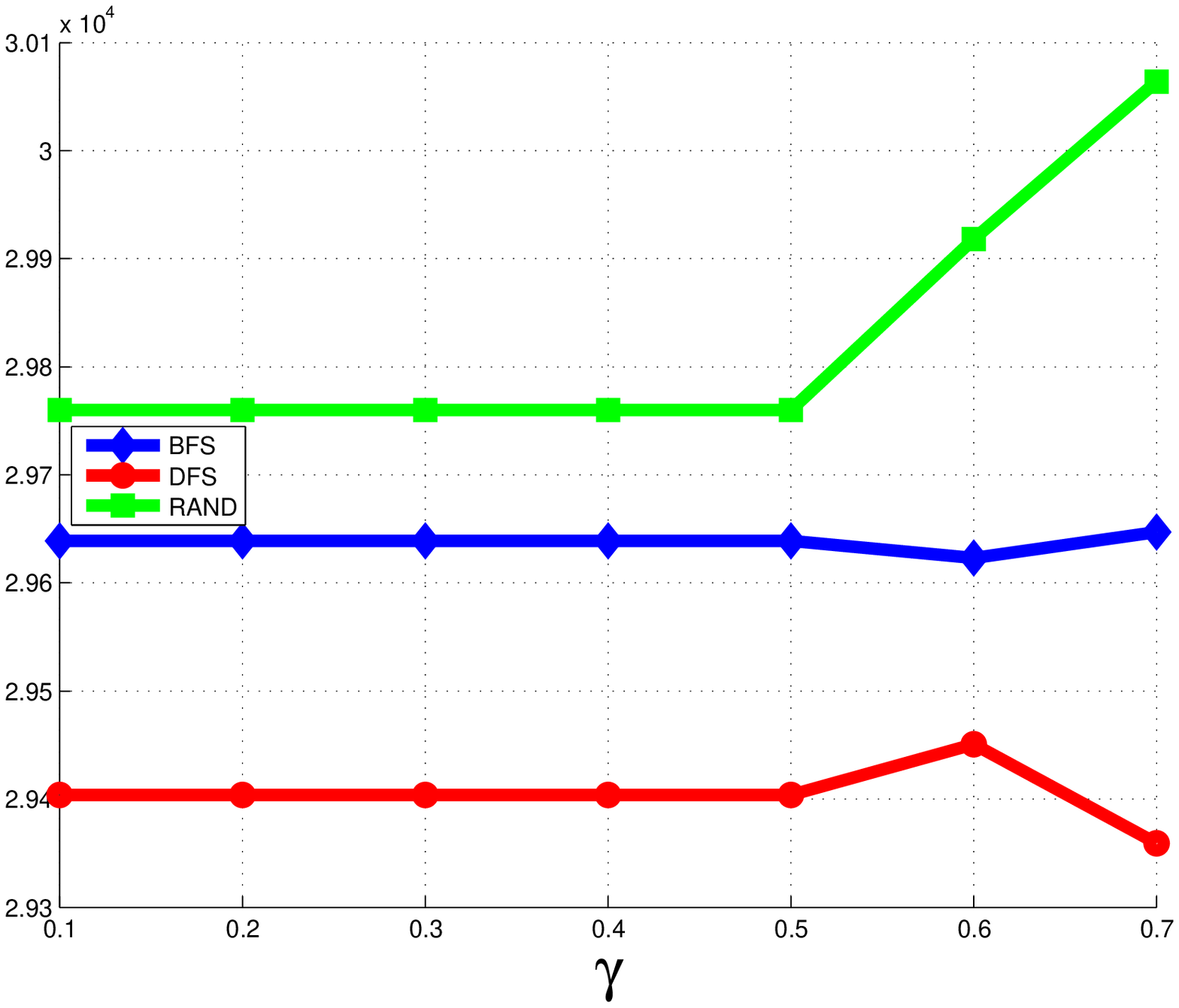}
\caption{$\beta = 0.0$}
\end{subfigure}
\begin{subfigure}[b]{0.48\textwidth}
\centering
\includegraphics[scale=0.35]{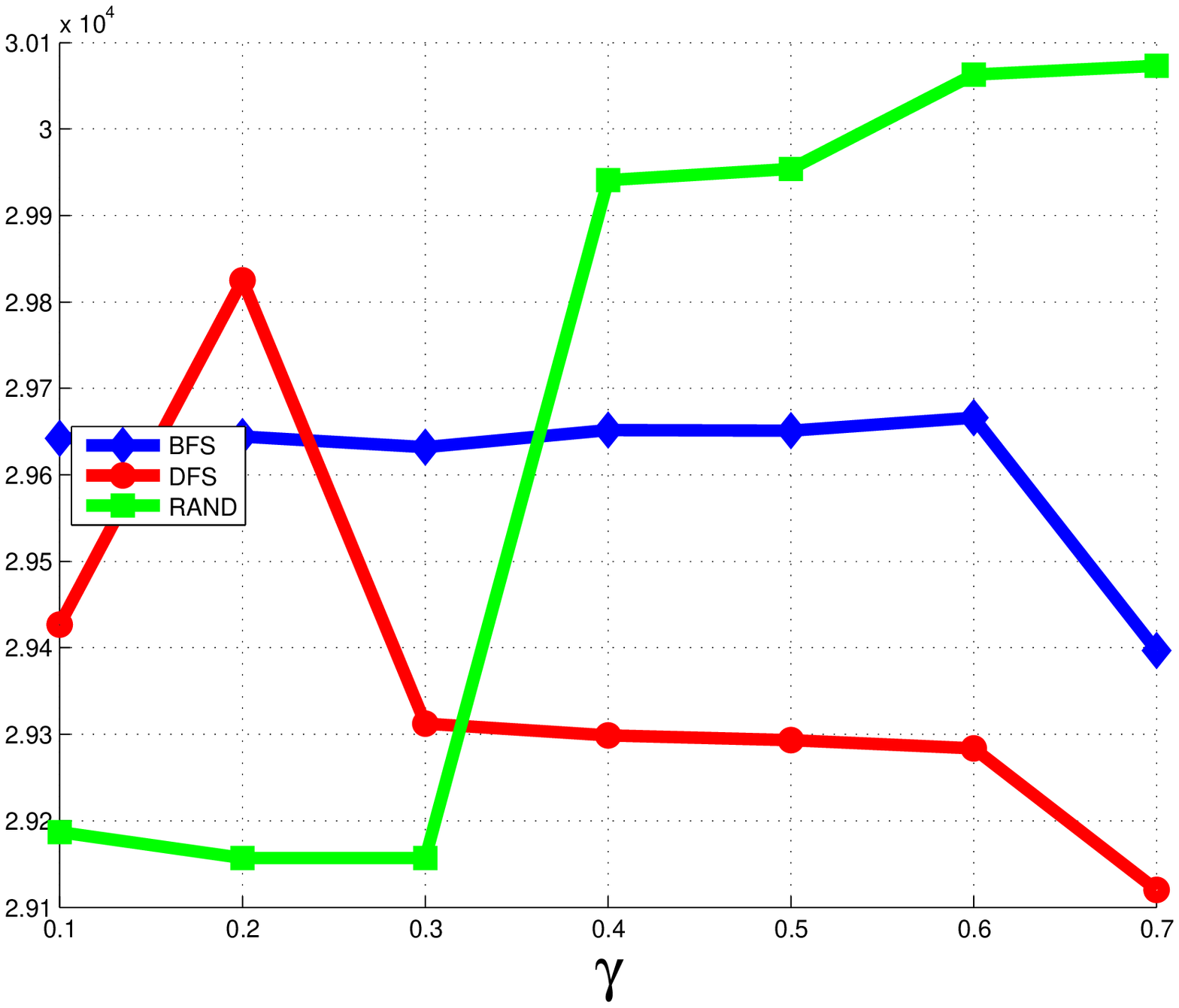}
\caption{$\beta = 0.1$}
\end{subfigure}
\begin{subfigure}[b]{0.48\textwidth}
\centering
\includegraphics[scale=0.35]{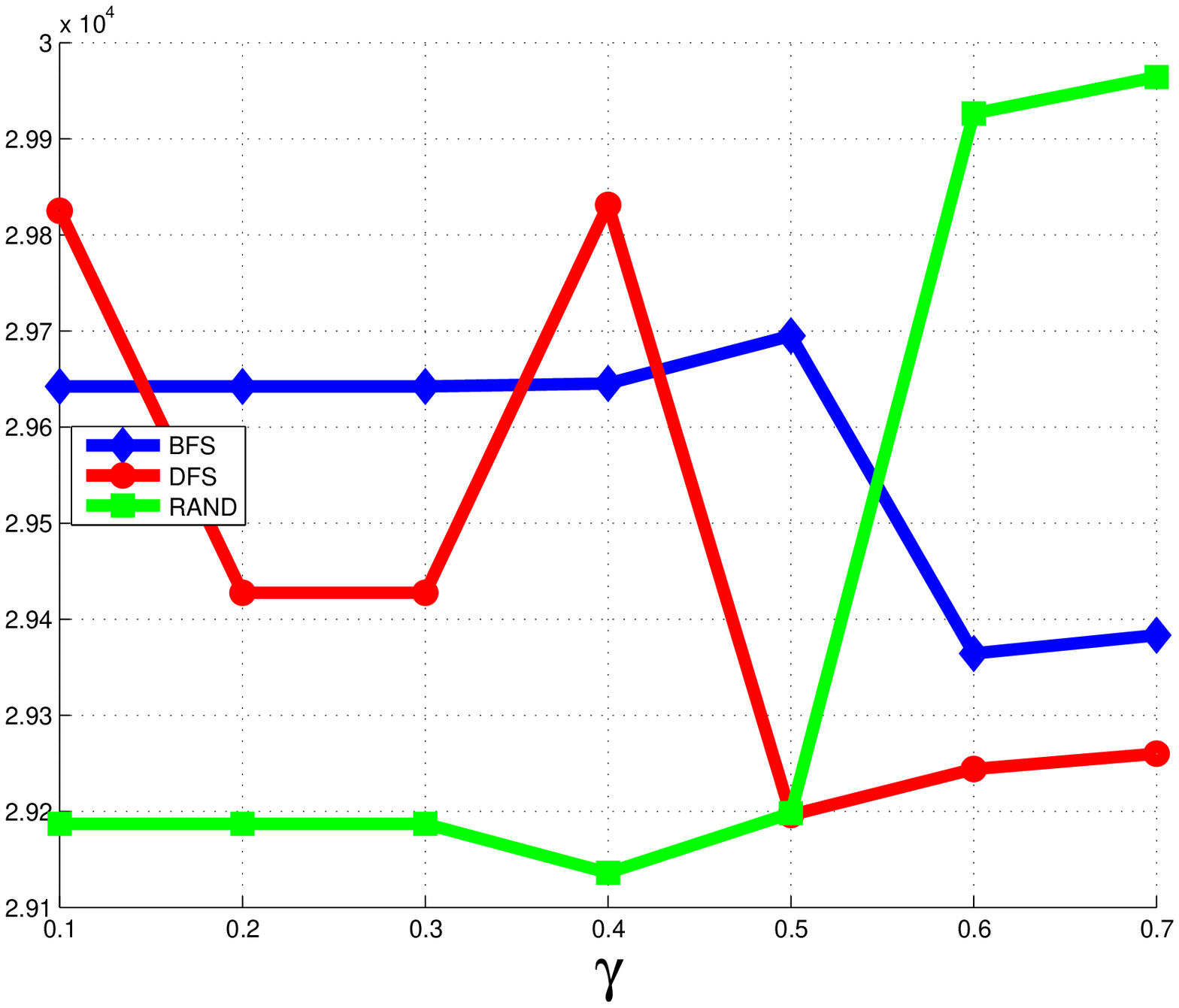}
\caption{$\beta = 0.2$}
\end{subfigure}
\begin{subfigure}[b]{0.48\textwidth}
\centering
\includegraphics[scale=0.35]{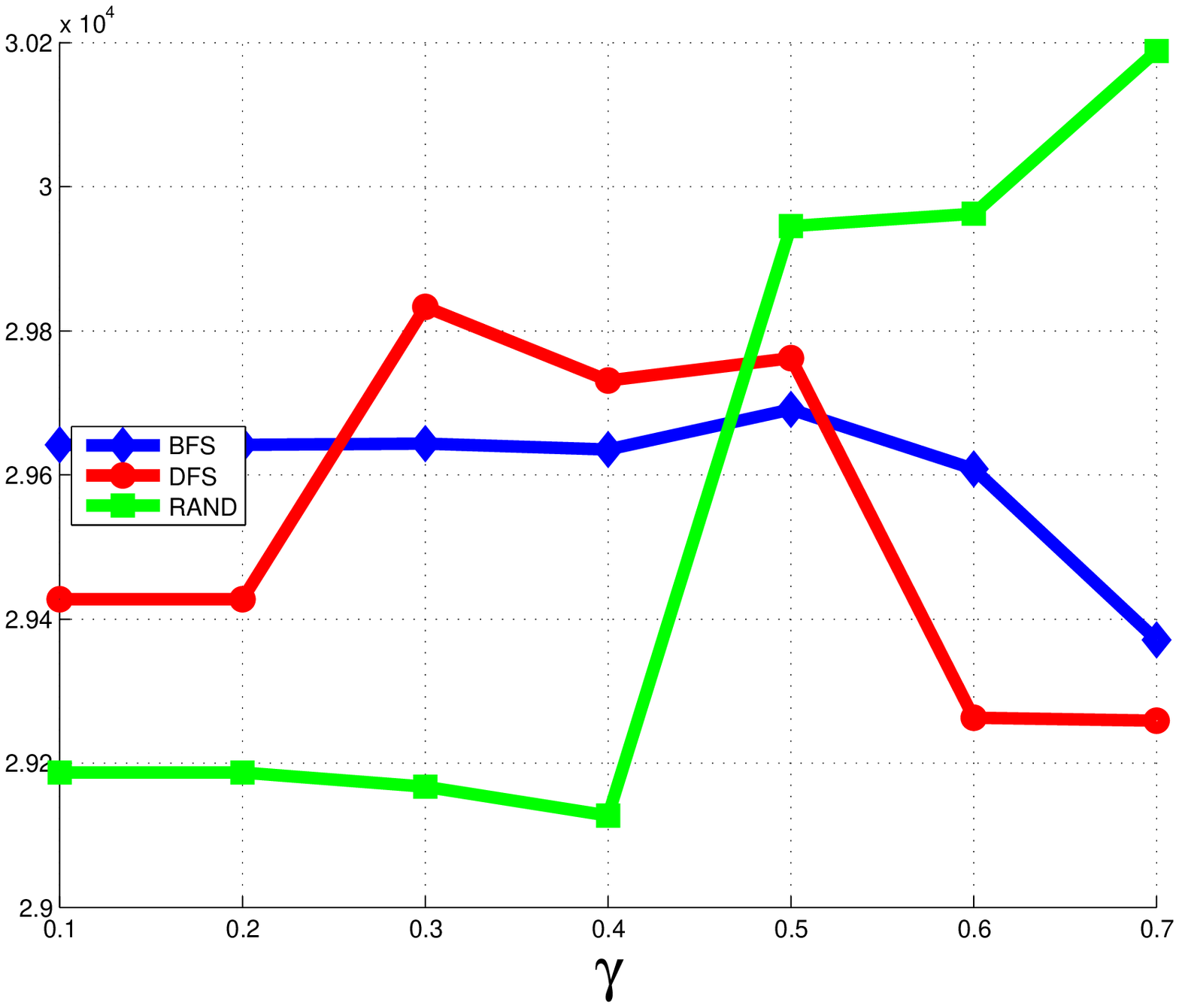}
\caption{$\beta = 0.3$}
\end{subfigure}
\caption{Via count vs. $\gamma$ for $n300$ and $\beta$ values: for Instance$\#1$}
\label{fig:gr-via-n300-bc}
\end{figure}

For the worst case instance, RAND gives smallest via count as compared to other modes for $\beta$ = $0.0$ and $\gamma \geq 0.4$ for $\beta > 0.0$. As $\beta$ increases, BFS dominates in lower values of $\gamma$, while RAND dominates for the remaining $\gamma$ values with fewer via counts. For all $\gamma$ values and the respective $\beta$ values, via count due to DFS mode is almost constant, with some variations near $\gamma$ value of $0.6$ and $0.7$.
\begin{figure}[!ht]
\centering
\begin{subfigure}[b]{0.48\textwidth}
\centering
\includegraphics[scale=0.35]{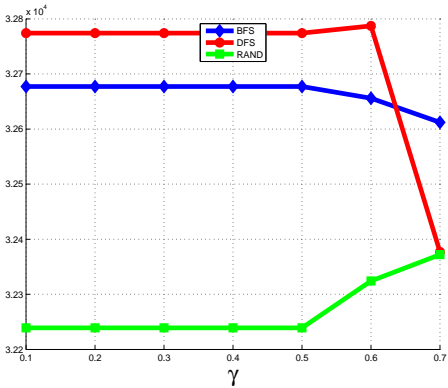}
\caption{$\beta = 0.0$}
\end{subfigure}
\begin{subfigure}[b]{0.48\textwidth}
\centering
\includegraphics[scale=0.35]{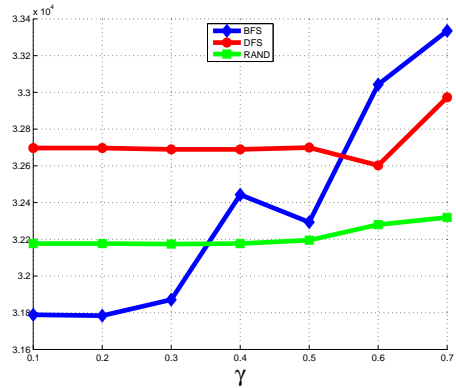}
\caption{$\beta = 0.1$}
\end{subfigure}
\begin{subfigure}[b]{0.48\textwidth}
\centering
\includegraphics[scale=0.35]{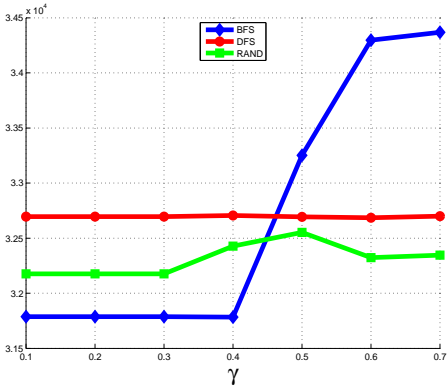}
\caption{$\beta = 0.2$}
\end{subfigure}
\begin{subfigure}[b]{0.48\textwidth}
\centering
\includegraphics[scale=0.35]{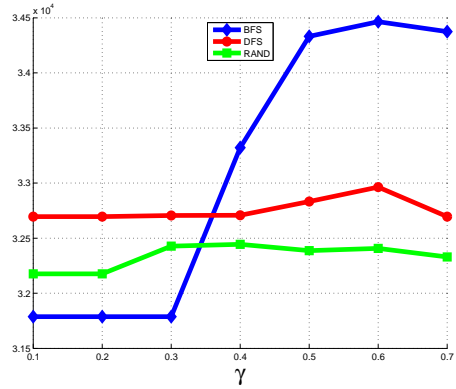}
\caption{$\beta = 0.3$}
\end{subfigure}
\caption{Via count vs. $\gamma$ for $n300$ and $\beta$ values: for Instance$\#2$}
\label{fig:gr-via-n300-wc}
\end{figure}

The experiments on all benchmark circuits for different floorplan instances showed that there was no significant variation in routed netlength obtained for BFS, DFS and RAND modes, but are better than that obtained in BFS-NB mode. Due to lack of space, we are not able to put the relevant details obtained by STAIRoute. These netlength values as normalized with resepect to no-blockage aware steiner length (computed by FLUTE \cite{cchu}) ratio and their geometric mean values were obtained as $1.207$, $1.201$ and $1.208$ for BFS, DFS and RAND modes respectively, while BFS-NB mode yields a value of $1.287$. Using the approach in \cite{weiy}, the average worst case congestion, defined as the ratio of routing demand and routing capacity, for different floorplan instances of all the circuits in all the modes and for all ($\gamma, \beta$) pairs, remained $~85\%$ ensuring $100\%$ routability, using up to $8$ metal layers as per the congestion model proposed in STAIRoute \cite{karb2}. However, the maximum average congestion \cite{weiy} in any of the floorplan instances for any mode and ($\gamma, \beta$) values was seen to be $~99\%$. This shows that no monotone staircase routing region had a congestion over $100\%$ in any routing layer as claimed by \cite{karb2}.

%%Section: Conclusion
\section{Conclusion}
\label{sec:con}
In this paper, we proposed an early version of unconstrained via minimization in floorplan based early global routing, by a new recursive floorplan bipartitioning framework. This bipartitioning framework identifies, for a given floorplan topology, a set of monotone staircase routing regions with minimal number of bends, by: (a) a greedy method employing BFS/DFS based graph search techniques, and (b) a randomized neighbor search technique for staircase wavefront propagation on BAG of the given floorplan. In this work, we first introduce the bend minimization objective in the multi-objective floorplan bipartitioning problem using monotone staircase cuts and used a pair of trade-off parameters ($\gamma, \beta$). The solution of this optimization yields a minimal bend monotone staircase routing which impacts the via count during floorplan based early global routing. 

Experimental results show the impact of the results of the proposed minimal bend monotone staircase bipartitioning methods on via count during early global routing for varying ($\gamma, \beta$) pairs and yield fewer via counts. This framework can potentially assess the quality of the floorplan in terms of these via counts. 

%The bipartitioning results also show that further improvement is possible with the randomized technique in order to obtain an early global routing solution with fewer via counts.
% \section*{Acknowledgment}
% The authors would also like to thank numerous reviewers for their valuable feedback.

%\appendix
%\appendixname

%\begin{appendices}
\section{Appendix}

\subsection{Staircase Wave-front Propagation in a Floorplan}
\label{apdx:1}
We consider an example of monotone staircase wave-front propagation in a floorplan instance for $9$ blocks, as depicted in Fig. \ref{fig:msc-wave-apdx}.  In this example, we study how different monotone staircase cuts on BAG can sequentially be obtained by the proposed DFS, BFS and randomized bipartitioning methods (RAND). This helps in exploring different sequences of monotone staircases with increased solution space for identifying an optimal monotone staircase for a given ($\gamma, \beta$) pair.
\begin{figure*}[!ht] %[H]
\centering
\includegraphics[scale=0.22]{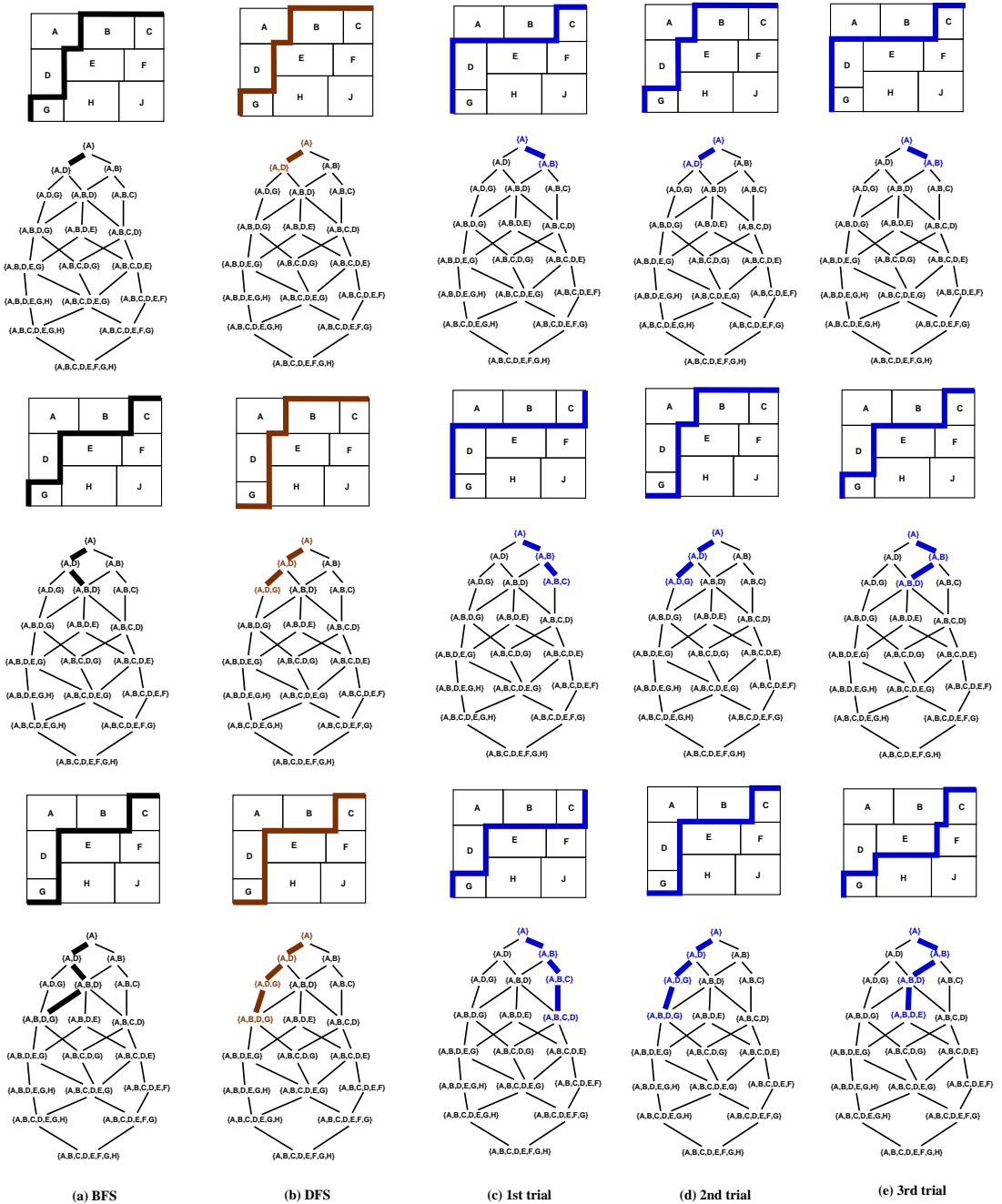}
\caption{Illustrating initial trail of sequences of monotone staircase wave-fronts: (a) BFS, (b) DFS, and (c) 1st, (d) 2nd, and (e) 3rd trial of RAND}
\label{fig:msc-wave-apdx}
\end{figure*}

Due to space limitation, only first few steps for identifying a sequence of monotone staircases obtained by BFS/DFS based bipartitioning are illustrated in Fig. \ref{fig:msc-wave-apdx} (a) and (b). It shows that both the methods greedily search the neighborhood of a vertex (block) in the BAG $G_b$ (see Fig. \ref{fig:neigh} (a)) for propagating the respective wave-fronts. Fig. \ref{fig:msc-wave-apdx} (c)-(e), illustrates three different trials of Algorithm \ref{alg:mscutrand} employing the proposed randomized neighbor search (see Fig. \ref{fig:neigh} (b)). The trials in RAND yield different wave-front propagation instances, as monotone staircase cuts on the BAG. It is important to note that BFS/DFS explores a fixed sequence of $9$ distinct staircases (see Lemma in \cite{karb}) for the same ($\gamma$, $\beta$) value, irrespective of the number of trials. On the other hand, RAND yields different sequences during different trials, by the proposed random neighbor indexing of the vertices. It is not necessary for the sequences to be fully disjoint as evident from Fig. \ref{fig:msc-wave-apdx} (c)-(e). Despite that, an increased solution space of different monotone staircases (a union of all of them obtained during different trials in RAND mode) facilitates us to identify an optimal monotone staircase with minimal number of bends for a given ($\gamma, \beta$), implied by the maximum $Gain$ value.

\subsection{Potential cross-talk minimization}
\label{apdx:2}
 
 This part discusses potential cross-talk minimization by minimizing the number of cut nets at any level of the bipartition hierarchy, MSC tree, by suitably choosing ($\gamma, \beta$) pair, as illustrated in Fig. \ref{fig:msc-cross-apdx}. In this example, we consider two instances of monotone staircases: (a) with more bends and net cut, and (b) with less bend and net cut, as depicted in Fig. \ref{fig:msc-cross-apdx} (a) and (b) respectively. In the former case, we see that two nets $a$ and $b$ are routed through the same monotone staircase routing region (MIS here) using same metal layer and therefore may results in signal cross talk among themselves. The latter case, however, shows that two different staircases are used to route nets $a$ and $b$; although net $b$ partly uses the same staircase (MIS), rest of its routing is done through a different staircase (MDS here). Therefore, both $a$ and $b$ will have minimal scope of signal interference between them.
 
 \begin{figure}[!ht] %[H]
 \centering
 \includegraphics[scale=0.45]{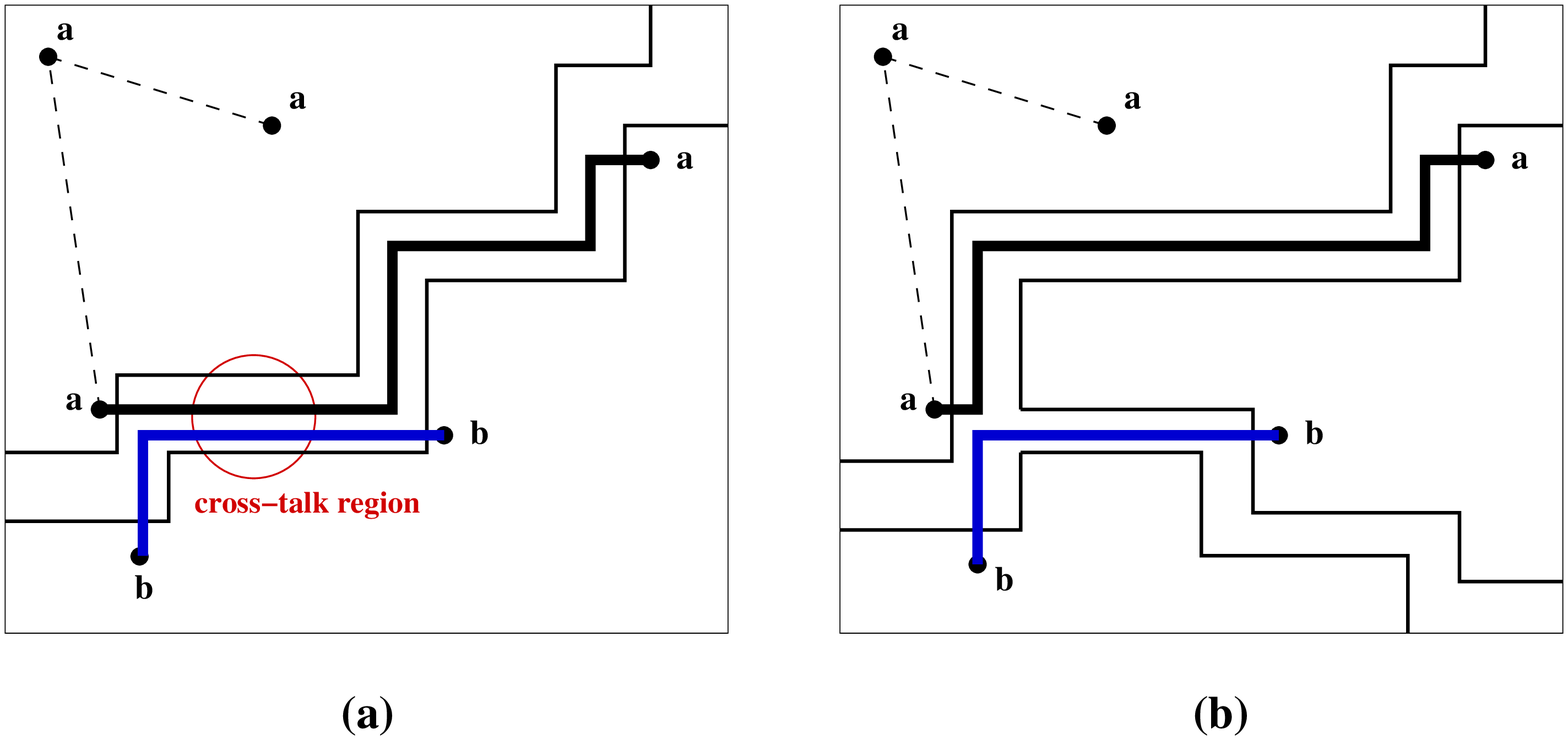}
 \caption{Scope of cross-talk between a pair of nets $a$ and $b$ with respect to minimal bends and cut nets in a monotone staircase with: (a) with more bends and cut nets,  and (b) less bends and cut nets}
 \label{fig:msc-cross-apdx}
 \end{figure}
%\end{appendices}

%% References
\renewcommand\bibname{References}
%\fancyhead[RE]{\bfseries References}
%\fancyhead[LO]{\bfseries References}
%\addcontentsline{toc}{chapter}{\bibname}
%\Urlmuskip=0mu plus 1mu
%\bibliographystyle{ACM-Reference-Format}
%\bibliography{BibFiles/references_ieeetran,BibFiles/my_publications_ieeetran}


\end{document}